\documentclass[aos]{imsart_v2}

\usepackage{amsthm,amsmath,natbib,graphicx,amssymb,float,amstext,bm,epstopdf,enumerate,xcolor,subcaption,multirow,xr}
\usepackage[ruled]{algorithm2e}
\RequirePackage[colorlinks,citecolor=blue,urlcolor=blue]{hyperref}
\usepackage{mathtools}
\usepackage{threeparttable}
\usepackage{longtable}
\usepackage{booktabs} 

\startlocaldefs

\numberwithin{equation}{section}

\newtheorem{Theorem}{theorem}
\newtheorem{Lemma}{Lemma}
\newtheorem{Corollary}{Corollary}
\newtheorem{Proposition}{Proposition}
\newtheorem{Example}{Example}

\definecolor{DarkRed}{rgb}{.7,0,.4}
\def\red{\textcolor{DarkRed}}

\def\bco{\iffalse}
\def\Cov{{\rm Cov}}
\def\Var{{\rm Var}}
\def\ci{\cite}
\def\cp{\citep}
\DeclareMathOperator*{\argmin}{argmin}

\def\red{\textcolor{red}}

\def\inv{^{-1}}
\def\d{\mathrm{d}}
\def\mc{\mathcal}
\def\dw{d_W}
\def\dinf{d_\infty}
\def\mk{\mathfrak}
\def\ast{^*}
\def\son{\sum_{i = 1}^n}
\DeclareMathOperator*{\esssup}{ess\,sup}
\def\i01{\int_0^1}

\def\xbar{\bar{X}}
\def\ybar{\bar{Y}}
\def\muy{\mu_Y}
\def\muz{\mu_Z}
\def\R{\mathbb{R}}
\def\fp{f_\oplus}
\def\hfp{\hat{f}_\oplus}
\def\Fp{F_\oplus}
\def\hFp{\hat{F}_\oplus}
\def\Qp{Q_\oplus}
\def\hQp{\hat{Q}_\oplus}
\def\qp{q_\oplus}
\def\hqp{\hat{q}_\oplus}
\def\hSigma{\hat{\Sigma}}

\def\D{\mc{D}}

\def\Vp{\Var_\oplus}
\def\Rp{\R^p}

\def\I{\mc{I}}
\def\Ist{{\I\ast}}
\def\Istc{{\bar{\I}\ast}}
\def\id{\mathrm{id}}

\def\T{^\top}

\def\zo{[0,1]}

\def\hlambda{\hat{\lambda}}
\def\hLambda{\hat{\Lambda}}

\def\Sigyz{\Sigma_{YZ}}
\def\Sigyy{\Sigma_{YY}}
\def\Sigzy{\Sigma_{ZY}}
\def\Sigzz{\Sigma_{ZZ}}

\def\hSigyy{\hat{\Sigma}_{YY}}
\def\hSigzy{\hat{\Sigma}_{ZY}}

\def\piy{\pi_Y}
\def\tgammay{\tilde{\gamma}_Y}

\def\tpiy{\tilde{\pi}_Y}

\def\bfp{f_{\oplus}\ast}
\def\bFp{F_\oplus\ast}
\def\bQp{Q_\oplus\ast}
\def\bqp{q_\oplus\ast}
\def\fzp{f_{0,\oplus}}
\def\sz{s_0}
\def\hfzp{\hat{f}_{0,\oplus}}
\def\hFzp{\hat{F}_{0,\oplus}}

\def\hbfp{\hat{f}_\oplus\ast}
\def\hbFp{\hat{F}_\oplus\ast}
\def\hbQp{\hat{Q}_\oplus\ast}

\def\hfi{\hat{\mk{F}}_i}

\def\hQi{\hat{Q}_i}

\def\hfzi{\hat{\mk{F}}_{0,i}}

\DeclarePairedDelimiterX{\normE}[1]{\lVert}{\rVert_E}{#1} 
\DeclarePairedDelimiterX{\normF}[1]{\lVert}{\rVert_F}{#1} 

\def\ipLt#1#2{\langle #1, #2\rangle_{L^2}}
\def\normLt#1{\lVert #1 \rVert_{L^2}}
\def\PX{P_{\mathbb{X}}}
\def\opx{o_{\PX}}
\def\Opx{O_{\PX}}
\def\EX{E_{\mathbb{X}}}
\def\VX{\Var_{\mathbb{X}}}
\def\rsa{\rightsquigarrow}
\def\CX{\Cov_{\mathbb{X}}}

\endlocaldefs

\begin{document}

\begin{frontmatter}

\title{Wasserstein $F$-tests and confidence bands for the Fr\'echet regression of density response curves}
\runtitle{Wasserstein Inference}


\author{\fnms{Alexander} \snm{Petersen}\thanksref{t1,m1}\ead[label=e1]{petersen@pstat.ucsb.edu}}
, 
\author{\fnms{Xi} \snm{Liu}\thanksref{m1}\ead[label=e2]{xliu@pstat.ucsb.edu}}
\and
\author{\fnms{Afshin A.} \snm{Divani}\thanksref{m2}\ead[label=e3]{adivani@salud.unm.edu}}
\affiliation{University of California, Santa Barbara\thanksmark{m1} and University of New Mexico\thanksmark{m2}}
\thankstext{t1}{Corresponding Author, supported by National Science Foundation grant DMS-1811888}

\address{Address of Alexander Petersen and Xi Liu\\
Department of Statistics and Applied Probability\\
University of California Santa Barbara\\
Santa Barbara, CA 93106-3110\\
\printead{e1}\\
\phantom{E-mail:\ }\printead*{e2}}

\address{Address of Afshin A. Divani \\
Department of Neurology \\
University of New Mexico \\
Albuquerque, NM 87131 \\
\printead{e3}}

\runauthor{Petersen, Liu and Divani}

\begin{abstract}
Data consisting of samples of probability density functions are increasingly prevalent, necessitating the development of methodologies for their analysis that respect the inherent nonlinearities associated with densities.  In many applications, density curves appear as functional response objects in a regression model with vector predictors.  For such models, inference is key to understand the importance of density-predictor relationships, and the uncertainty associated with the estimated conditional mean densities, defined as conditional Fr\'echet means under a suitable metric.  Using the Wasserstein geometry of optimal transport, we consider the Fr\'echet regression of density curve responses and develop tests for global and partial effects, as well as simultaneous confidence bands for estimated conditional mean densities.  The asymptotic behavior of these objects is based on underlying functional central limit theorems within Wasserstein space, and we demonstrate that they are asymptotically of the correct size and coverage, with uniformly strong consistency of the proposed tests under sequences of contiguous alternatives.  The accuracy of these methods, including nominal size, power, and coverage, is assessed through simulations, and their utility is illustrated through a regression analysis of post-intracerebral hemorrhage hematoma densities and their associations with a set of clinical and radiological covariates.
\end{abstract}



\end{frontmatter}

\section{Introduction}
\label{sec: intro}

Samples of probability density functions arise naturally in many modern data analysis settings, including population age and mortality distributions across different countries or regions \cp{mena:16:2,bigo:17,mull:19:3}, and distributions of functional connectivity patterns in the brain \cp{mull:16:1}.  Methods for the analysis of density data began with the work of \ci{knei:01}, who applied standard functional principal components analysis (FPCA) in order to quantify a mean density and prominent modes of variability about that mean.  However, due to inherent constraints of density functions, which must be nonnegative and integrate to one, nonlinear methods are steadily replacing such standard procedures.  For example, \ci{mena:16:2} and \ci{mull:16:1} both proposed to apply a preliminary transformation to the densities, mapping them into a Hilbert space, after which linear methods such as FPCA can be suitably applied.  The transformation of \ci{mena:16:2} was specifically motivated by the extension of the Aitchison geometry \cp{aitc:86} to infinite-dimensional compositional data due to the work of \ci{egoz:06}, while those in \ci{mull:16:1} were generic and not motivated by any particular geometry.

Parallel developments in the analysis of nonlinear data have been made in the broader field of object-oriented data analysis \cp{marr:14}, where a complex data space is endowed with a chosen metric that, in turn, defines the parameters of the model.  Particular attention has been paid to manifold-valued data (e.g., \ci{flet:04, pana:14}), with the main objects of interest being the Fr\'echet mean and variance, as well as dimension reduction tools designed to optimally retain variability in the data \cp{patr:15}.  Within this context, \ci{sriv:07} and \ci{bigo:17} developed manifold-based dimension reduction techniques specifically designed for samples of density functions, the former utilizing the Fisher-Rao geometry and the latter the Wasserstein geometry of optimal transport.  Of these two, the Wasserstein metric has proved in recent years to have greater appeal both theoretically, given its clear interpretation as an optimal transport cost \cp{vill:03,ambr:08}, as well as in applied settings \cp{bols:03,broa:06,mull:11:4,pana:16}.

In this paper, we study a regression model with density functions as response variables under the Wasserstein geometry, with predictors being Euclidean vectors.  Such data are frequently encountered in practice (e.g., \ci{neri:07,tals:18}). The goal of the model is to perform inference, specifically tests for covariate effects and confidence bands for the fitted conditional mean densities.  For this reason, we assume a global regression model that does not require any smoothing or other tuning parameter to fit.  Standard functional response models, such as the linear model \cp{fara:97}, are not suitable for the nonlinear space of density functions unless the densities are first transformed into a linear space, as demonstrated recently in \ci{tals:18} using the compositional approach.  However, there is no such transformation that is suitable for the Wasserstein geometry.  Global regression models for Riemannian manifolds have been developed \cp{flet:13,niet:11,corn:17}, but are also not directly applicable to the Wasserstein geometry.  Instead, we will develop our inferential techniques under the Fr\'echet regression model proposed in \ci{mull:19:3}, which defines a global regression function between response data in an arbitrary metric space and vector predictors.  Although the theory of estimation under this model was well-studied in the general metric space framework, this is not the case for other forms of inference.  

\begin{figure}[t]
\label{fig: ObsFitted}
\centering
\includegraphics{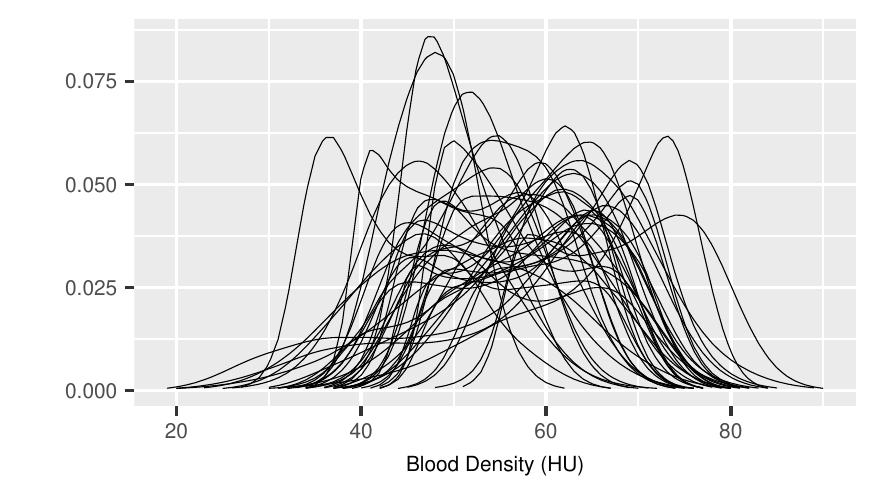} 
\caption{\footnotesize Observed hematoma density distributions for a randomly selected subset of 40 post-intracerebral hemmorhage patients.\label{fig: ObsFitted}}
\end{figure}
\normalsize

In Section~\ref{sec: prelim}, we briefly describe the necessary components of the Wasserstein geometry and its implications when applied to the Fr\'echet regression model.  An additional component not considered in \ci{mull:19:3} that is available through this particular formulation is the random optimal transport map between the conditional Wasserstein mean density and the observed one.  These maps serve the purpose of the error term in the regression model, although they do not act additively or even linearly.  In Section~\ref{sec: tests}, we develop intuitive test statistics for covariate effects and derive their asymptotic distributions, leading to root-$n$ consistent testing procedures.  We also describe different methods for implementing these tests, including a bootstrap procedure involving residual optimal transport maps obtained from the fitted model.  Section~\ref{sec: cb} demonstrates how to compute asymptotic confidence bands about the estimated conditional mean distributions.  Finally, these methods are illustrated through extensive simulations (Section~\ref{sec: sims}) and the analysis of distributions of hematoma density for stroke patients (Section~\ref{sec: stroke}), as captured by computed tomography scans, and their dependency on patient-specific covariates (Figure~\ref{fig: ObsFitted}).

\section{Preliminaries}
\label{sec: prelim}

We begin with a brief definition of the Wasserstein metric in the language of optimal transport.   This Wasserstein metric has been known under other names, including ``Mallows'' and ``earth-mover's" \cp{levi:01}, and its use in statistics is rapidly expanding \cp{pana:19}.  Let $\mc{D}$ be the class of univariate probability density functions $f$ that satisfy $\int_{\R} u^2 f(u) \d u < \infty$, i.e. absolutely continuous distributions on with a finite second moment.  In fact, the Wasserstein metric is well-defined for such distributions without requiring a density, but for simplicity of presentation we deal specifically with this subclass.  For a comprehensive treatment in the more general case, see \ci{vill:03} or \ci{ambr:08}.  For $f,$ $g \in \D,$ consider the collection of maps $\mc{M}_{f,g}$ such that, if $U \sim f$ and $M \in \mc{M}_{f,g}$, then $M(U) \sim g.$  The squared Wasserstein distance between these two distributions is
$$
\dw^2(f, g) = \inf_{M \in \mc{M}_{f,g}} \int_\R \left(M(u) - u\right)^2 f(u)\d u.
$$
It is well known that the infimum above is attained by the optimal transport map $M_{f,g}^{\textrm{opt}} = G\inv \circ F$, where $F$ and $G$ are the cumulative distribution functions of $f$ and $g$, respectively, leading to the closed forms
\begin{equation}
\label{eq: wass2}
\dw^2(f, g) = \int_\R \left(M_{f,g}^{\textrm{opt}}(u) - u\right)^2 f(u) \d u  = \int_0^1 \left(F\inv(t) - G\inv(t)\right)^2\d t,
\end{equation}
where the last equality follows by the change of variables $t = F(u).$  A more proper term for this metric is the Wasserstein-2 distance, since it is just one among an entire class of Wasserstein metrics.  

Within this larger class of metrics is also the Wasserstein-$\infty$ metric, which will be useful in the formation of confidence bands.  For two densities $f,g \in \D,$ their Wasserstein-$\infty$ distance is
\begin{equation}
\label{eq: wassInf}
\dinf(f, g) = f\textrm{-}\esssup_{u \in \I} |M_{f,g}^{\textrm{opt}}(u) - u|, 
\end{equation}
where $f\textrm{-}\esssup$ refers to the essential supremum with respect to $f$.  

\subsection{Random Densities}
\label{ss: randomDens}

A random density $\mk{F}$ is a random variable taking its values almost surely in $\mc{D}.$  It will also be useful to refer to the corresponding random cdf $F$, quantile function $Q = F\inv$ and quantile density $q = Q' = 1/(\mk{F} \circ Q)$ \cp{parz:79}.  For clarity, $u,v \in \R$ will consistently be used as density and cdf arguments throughout, whereas $s, t \in[0,1]$ will be used as arguments for the quantile and quantile density functions.  Following the ideas of \ci{frec:48}, the Wasserstein--Fr\'echet (or simply Wasserstein) mean and variance of $\mk{F}$ are
\begin{equation}
\label{eq: wMeanVar}
\bfp := \argmin_{f \in \D} E\left(\dw^2(\mk{F}, f)\right), \quad \Var_\oplus(\mk{F}):= E\left(\dw^2(\mk{F}, \bfp)\right).
\end{equation}

In the regression setting, we model the distribution of $\mk{F}$ conditional on a vector $X \in \R^p$ of predictors, where the pair $(X, \mk{F})$ is distributed according to a probability measure $\mc{G}$ on the product space $\R^p \times \mc{D}.$  In this sense, the objects in \eqref{eq: wMeanVar} are the marginal Fr\'echet mean and variance of $\mk{F}.$    Let $\mk{S}_X$ denote the support of the marginal distribution of $X.$ Our interest is in the Fr\'echet regression function, or function of conditional Fr\'echet means,
\begin{equation}
\label{eq: condWMean}
\fp(x) := \argmin_{f \in \D} E\left[\dw^2(\mk{F}, f) | X = x\right], \quad x \in \mk{S}_X.
\end{equation}
Let $\Fp(x)$, $\Qp(x)$, and $\qp(x)$ denote, respectively, the cdf, quantile, and quantile density functions corresponding to $\fp(x).$  We will use the notation $\fp(x,u)$ to denote the value of the conditional mean density $\fp(x)$ at argument $u \in \R,$ and similary for $\Fp(x,u),$ $\Qp(x,t)$, and $\qp(x,t),$ $t \in \zo.$ For a pair $(X,\mk{F})$, define $T = Q \circ \Fp(X)$ to be the optimal transport map from the conditional mean $\fp(X)$ to the random density $\mk{F}.$  By \eqref{eq: wass2}, it must be that $E(Q(t)|X=x) = \Qp(x,t),$ so that $E(T(u) |X = x) = u$ for $u$ such that $\fp(x,u) > 0.$ Then the conditional Fr\'echet variance is
\begin{equation}
\label{eq: condWVar}
\begin{split}
\Vp(\mk{F}|X= x) &= E\left[\dw^2(\mk{F}, \fp(x))|X = x\right] = \int_\R E\left[(T(u) - u)^2|X=x\right]\fp(x,u)\d u \\
&= \int_\R \Var(T(u)|X=x) \fp(x,u) \d u.
\end{split}
\end{equation}

In these developments, we have assumed that the marginal and conditional Wasserstein mean densities exist and are unique.  However, this is not automatic and some conditions are needed.  Previous work on existence, uniqueness, and regularity of Wasserstein means, or barycenters, for a finite collection of probability meaures on $\R^d$ was done by \ci{ague:11}, and extended to continuously-indexed measures by \ci{pass:13}.  For random probability measures with support on $\zo,$ \ci{pana:16} (see Proposition 2 therein) gave sufficient conditions for the existence and uniqueness of a Wasserstein mean measure, although it was not guaranteed to have a density.  However, none of these are sufficiently strong for the purposes of this paper, where existence, uniqueness, and regularity of both marginal \eqref{eq: wMeanVar} and conditional \eqref{eq: condWMean} Wasserstein means are needed.  To this end, consider the following assumptions on the joint distribution $\mc{G}.$
\begin{enumerate}
\item[(A1)] $\mk{F}(u) \in (0,\infty)$ for $u \in (Q(0),Q(1))$ almost surely, $\Var(Q(t))< \infty$ for all $t\in (0,1),$ and $\i01 \Var(Q(t)) \d t < \infty.$
\item[(A2)]  For any $t \in (0,1),$ there exists $\delta,$ $0 < \delta < \min\{t,1-t\}$ such that $E\left(\sup_{|s-t|<\delta} q(s)\right)< \infty.$
\item[(A3)] For all $x \in \mk{S}_X,$ $P\left(\sup_{u \in (Q(0), Q(1))} \mk{F}(u) < \infty | X = x\right) > 0.$
\end{enumerate}

Assumption (A1) is essentially the same as that made in Proposition 2 of \ci{pana:16}, with additional moment assumptions on $Q$ since  $\mk{F}$ is not assumed to be supported on any bounded interval, and implies existence and uniqueness of the Wasserstein mean measures.  Assumption (A2) is a regularity condition to ensure these Wasserstein means have densities in $\mc{D},$ and (A3) implies that these mean densities are bounded;  see \ci{pass:13} for a similar assumption.  The proof of the following and all other theoretical results can be found in the Appendix.
\begin{Proposition}
\label{prop: wMeanUnique}
Suppose conditions (A1)--(A3) hold.  Then the marginal Wasserstein mean $\bfp$ in \eqref{eq: wMeanVar} exists, is a unique element of $\mc{D},$ and satisfies $\sup_{u} \bfp(u) < \infty$ and \mbox{$\Vp(\mk{F}) < \infty.$}  For all $x \in \mk{S}_X,$ the conditional Wasserstein mean $\fp(x)$ in \eqref{eq: condWMean} exists, is a unique element of $\mc{D}$, and satisfies $\sup_{u}\fp(x,u) < \infty$ and $\Vp(\mk{F}|X = x) < \infty.$
\end{Proposition}

\subsection{Global Wasserstein--Fr\'echet Regression}
\label{ss: regModel}

In order to facilitate inference, specifically tests for no or partial effects of the covariates $X$ and confidence bands for the conditional Wasserstein means, we consider a particular global regression model for the conditional Wassersteins $\fp(x)$ defined in \eqref{eq: condWMean}.  This model, proposed by \ci{mull:19:3}, is termed Fr\'echet regression, and takes the form of a weighted Fr\'echet mean
\begin{equation}
\label{eq: Wmodel}
\fp(x) = \argmin_{f \in \D} E\left[s(X, x)\dw^2(\mk{F}, f)\right],
\end{equation}
where the weight function is
$$
s(X, x) = 1 + (X - \mu)^\top \Sigma\inv(x - \mu), \quad \mu = E(X), \, \Sigma = \Var(X),
$$ 
and $\Sigma$ is assumed to be positive definite.  The model is motivated by multiple linear regression, and is its direct generalization to the case of a response variable in a metric space.  Specifically, if a scalar response $Y$ is jointly distributed with $X$ and $E(Y|X = x)$ is linear in $x,$ \ci{mull:19:3} showed that an alternative characterization of linear regression is
$$
E(Y|X = x) = \argmin_{y \in \R} E\left[s(X,x)(Y-y)^2\right].
$$
Thus, model \eqref{eq: Wmodel} generalizes linear regression to the case of density response by substituting $\mk{F}$ for $Y$ and $(\mc{D},d_W)$ in place of the usual metric space $(\R, |\cdot|)$ implicitly used in multiple linear regression.  Although $\mc{D}$ is not a linear space, \eqref{eq: Wmodel} provides a sensible regression model for the conditional Wasserstein means that retains some properties of linear regression.  For example, since $s(z, \mu)\equiv 1,$ we have $\fp(\mu) = \bfp,$ so that the regression function passes through the point $(\mu, \bfp).$

For the remainder of the paper, and implicit in the statement of all theoretical results, we assume that the distribution $\mc{G}$ satisfies model \eqref{eq: Wmodel}, with $\bfp,$ $\fp(x)$ being unique elements of $\mc{D}.$  We now give a very basic example of this model, motivated by the well-known connection between the Wasserstein metric and location-scale families (e.g., \ci{bigo:17}).
\begin{Example}
\label{exm: LocScale}
Suppose $X \in \R$ and fix $f_0 \in \mc{D}.$ Letting $a_j, b_j \in \R$, define 
$$
\nu(x) = a_0 + a_1x, \quad \tau(x) = b_0 + b_1x,
$$
and suppose $\tau(X) > 0$ almost surely.  Let $V_1,V_2$ satisfy $E(V_1|X)= 0,$ $E(V_2|X) = 1,$ and $V_2 > 0$ almost surely. Then the location-scale model
$$
\mk{F}(u) = \frac{1}{V_2\tau(X)}f_0\left(\frac{u - V_1 - V_2\nu(X)}{V_2\tau(X)}\right)
$$
corresponds to \eqref{eq: Wmodel} with $$\fp(x,u) = \frac{1}{\tau(x)} f_0\left(\frac{u - \nu(x)}{\tau(x)}\right).$$

To show the above, it is sufficient to show that $E(Q(t)|X) = \Qp(X,t)$ almost surely.  Let $F_0$ be the cdf corresponding to $f_0,$ so that $Q_\oplus(x,t) = \nu(x) + \tau(x)Q_0(t).$ Since
$$
Q(t) = V_1 + V_2\nu(X) + V_2\tau(X)Q_0(t),
$$
it is easily verified that $E(Q(t)|X) = \Qp(X,t)$ by the properties of $V_1,$ $V_2.$  Moreover, note that the optimal transport map from $\fp(X)$ to $\mk{F}$ is
$$
T(u) = Q \circ \Fp(X,u) = V_1 + V_2u,
$$
and satisfies $E(T(u)|X) = u$ almost surely.
\end{Example}

Thus far, the regression model \eqref{eq: Wmodel} provides us with a formula for the conditional Wasserstein mean of $\mk{F},$ whereas one also needs information on the conditional variance in order to conduct inference.  To this end, observe that $Q = T \circ \Qp(X),$ so that $T$ acts as a residual transport, although it acts on the quantile function, and not the density, and does so through composition and not additively.  As oberved in Section~\ref{ss: randomDens}, the first order behavior of the residual transport $T$ is completely specified by the model as $E(T(u)|X) = u$ almost surely.  We also impose the following assumption on the covariance.
\begin{enumerate}
\item[(T1)] The covariance function 
$
C_T(u, v) = \Cov(T(u), T(v))
$
is continuous, $\sup_{u \in \R} \Var(T(u)) < \infty,$ and $\Cov(T(u), T(v) | X) = C_T(u, v)$ almost surely.
\end{enumerate}
This corresponds to the classical constant variance or exogeneity assumption requiring that the second order behavior of the residual transport be independent of the predictors. Define \mbox{$S = \Qp(X) \circ \bFp$}, which is the optimal transport map from the marginal Wasserstein mean to the conditional one.  Again, it is easily verified that \mbox{$E(T\circ S(u)) = u$} for $u$ such that $\bfp(u) > 0.$ These observations lead to the following decomposition of the Wasserstein variance.
\begin{Proposition}
\label{prop: wanova}
Suppose that assumption (T1) is satisfied.  Then
\begin{equation}
\label{eq: wanova}
\begin{split}
\Vp(\mk{F}) &= E\left[\Vp(\mk{F}|X)\right] + \Vp(\fp(X)) \\
&= \int_\R E\left[C_T(S(u), S(u))\right] \bfp(u) \d u + \int_\R(S(u) - u)^2\bfp(u) \d u.
\end{split}
\end{equation}
\end{Proposition}
While a standard result in Euclidean spaces, the above variance decomposition does not generally hold in metric spaces such as $\mc{D}.$  However, Proposition~\ref{prop: wanova} demonstrates that this decomposition does indeed hold for random densities under the Wasserstein metric whenever model \eqref{eq: Wmodel} holds.  This finding motivates the specific choices of test statistics developed in Section~\ref{sec: tests}.

\subsection{Estimation}
\label{ss: est}

In order to estimate the regression function $\fp(x),$ we utilize an empirical version of the least-squares Wasserstein criterion in \eqref{eq: Wmodel}.   First, set  \mbox{$\xbar = n\inv\son X_i,$} \mbox{$\hSigma = n\inv\son (X_i - \xbar)(X_i - \xbar)^\top,$} and compute empirical weights
$
s_{in}(x) = 1 + (X_i - \bar{X})^\top \hSigma\inv(x - \bar{X}),
$
Let $\mk{Q}$ be the set of quantile functions in $L^2\zo.$  With $\normLt{\cdot}$ denoting the standard Hilbert norm on $L^2\zo$, an estimator of $\Qp(x)$ is
\begin{equation}
\label{eq: Qfit}
\hQp(x) = \argmin_{Q \in \mk{Q}} \son s_{in}(x) \normLt{Q - Q_i}^2.
\end{equation}
Implementation of this estimator is given in Algorithm~\ref{alg: Qest} in the Appendix.  In finite samples, $\hQp(x)$ will not necessarily admit a density.  For almost all procedures described, this quantile estimate is sufficient, since it can be used to compute Wasserstein distances as in \eqref{eq: wass2}.  Nevertheless, we will establish (see Lemma~\ref{lma: regular} in the Appendix) that $\hQp(x)$ admits a density $\hfp(x)\in \mc{D}$ for large samples with high probability. When this holds, the estimate
\begin{equation}
\label{eq: fullEst}
\hfp(x) = \argmin_{f \in \D} \frac{1}{n}\son s_{in}(x)\dw^2(\mk{F}_i, f).
\end{equation}
is well-defined, and $\hfp(x)$ is the density corresponding to the quantile estimate $\hQp(x)$ above.  It can be computed in practice using Algorithm~\ref{alg: fest} in the Appendix.

Since we will also consider hypothesis tests of partial effects, write $x \in \Rp$ as $x = (y,z),$ where $y$ corresponds to the first $q$ entries of $x,$ $q < p.$  Similarly, take $X = (Y, Z),$ $\mu = (\muy, \muz),$ and
$$
\Sigma = \left(\begin{array}{c c} \Sigyy & \Sigyz \\ \Sigzy & \Sigzz \end{array}\right),
$$
with corresponding notations for the partitions of $\xbar$ and $\hSigma.$ Consider the null hypothesis $\mc{H}_0^P: \fp(x) = \fzp(y),$
$$ 
\fzp(y) := \argmin_{f \in \mc{D}}E\left[\sz(Y, y)\dw^2(\mk{F}_i, f)\right], \,\, \sz(Y, y)=1 + (Y - \mu_Y)\Sigyy\inv(y - \muy).
$$
Then the restricted estimators $\hat{Q}_{0,\oplus}(y)$ and $\hfzp(y)$ are defined analogously to \eqref{eq: Qfit} and \eqref{eq: fullEst}, only using submodel weights 
$$
s_{in,0}(y) = 1 + (Y - \ybar)\T \hSigyy\inv (y - \ybar).
$$

\section{Hypothesis Testing}
\label{sec: tests}

Once estimation is under control, the first goal of any global regression model is to test for effects of the predictors.  In the more abstract setting of a response variable in an arbitrary metric space, \ci{mull:19:3} suggested a permutation approach based on a Fr\'echet generalization of the coefficient of determination $R^2$ in multiple linear regression, though the theoretical properties of this test were not investigated. In a recent preprint, \ci{mull:18:5} developed a test statistic and its asymptotic distribution for the case of a random object response and categorical predictors, giving a Fr\'echet extension of analysis of variance.  Given that we are considering the more specific case of density-valued response variables under the Wasserstein geometry, we are able to expand on these results in order to develop asymptotically justified tests for both global and partial null hypotheses, where predictors can be of any type.

\subsection{Test of No Effects}
\label{ss: global}

We begin with the global null hypothesis of no effects, $\mc{H}_0^G: \fp(x) \equiv \bfp$.  Given the Wasserstein variance decomposition in \eqref{eq: wanova}, under $\mc{H}_0^G$ we have $\Vp(\fp(X)) = E\left(\dw^2(\fp(X), \bfp)\right) = 0.$
This motivates
\begin{equation}
\label{eq: globalF}
F\ast_G = \sum_{i = 1}^n \dw^2(\hfi, \hbfp)
\end{equation}
as a test statistic, where $\hfi = \hfp(X_i)$ are the fitted densities and $\hbfp = \hfp(\xbar)$ is the sample Wasserstein mean. This can be viewed as a generalization of the numerator of the global $F$-test in multiple linear regression, and we thus refer to $F\ast_G$ in \eqref{eq: globalF} as the (global) Wasserstein $F$-statistic.

In order to establish the asymptotic null distribution of $F\ast_G,$ we require the following assumptions.  Define 
$
C_T^{(l,m)} = \frac{\partial^{l+m}}{\partial u^l \partial v^m}C_T
$
and, for any $x \in \R^p,$ set $\tilde{x} = (1\,\, x\T)\T.$  Let $R = T\circ S = Q\circ \bFp$ and $\I_x = (\Qp(x,0), \Qp(x,1)).$
\begin{enumerate}
\item[(T2)] $T$ is differentiable almost surely, with $T'$ having random Lipschitz constant $L.$  The covariance function $C_T$ has continuous partial derivatives $C_T^{(l,m)},$ for $l,m=0,1,$ with $\sup_{u \in \R} C_T^{(1,1)}(u,u) < \infty.$
\item[(T3)] $E(\lVert X \rVert_E^4),$ $E(\lVert \tilde{X}\rVert_E^4 \sup_{u \in \R} |T'(u)|^2),$ and $E(\lVert \tilde{X}\rVert_E^6 L^2)$ are all finite.
\item[(T4)] $\Ist = \I_\mu$ is a bounded interval, and $\gamma(u) = \Cov(X, R(u))$ has bounded second derivative on $\Ist$.  Also, $\inf_{u \in \Ist} \bfp(u) > 0$ and, for some $M > 0,$ $\sup_{u \in \I_x} \fp(x,u) < M$ for almost all $x.$
\end{enumerate}
Assumptions (T2) and (T3) impose conditions on the joint distribution of $(X, T).$  Condition (T2) is a smoothness condition on the optimal transport process $T$, while the moment requirements in (T3) are tightness conditions that allow for the necessary asymptotic Gaussianity to hold; see Lemma~\ref{lma: Qjoint} in the Appendix.  Assumption (T4) involves the regression relationship between $X$ and $\mk{F};$ importantly, it ensures that the conditional mean densities are sufficiently and uniformly separated from the boundary of $\mc{D}$ within the space of distributions with finite second moments.  In the spirit of regression, we will consider the asymptotic behavior of $F\ast_G$ conditional on the observed predictors. Define the covariance kernel 
\begin{equation}
\label{eq: Keigen}
K(u,v) = E\left[C_T(S(u), S(v))\right] = \sum_{j = 1}^\infty \lambda_j \phi_j(u)\phi_j(v), \quad u, v \in \Istc,
\end{equation}
where $\Istc$ is the closure of $\Ist.$  The right-hand side is the Mercer decomposition of $K$ \cp{hsin:15}, so that $\{\phi_j\}_{j = 1}^\infty$ forms an orthonormal set in $L^2(\Istc, \bfp)$ and $\lambda_j$ are positive, nonincreasing in $j$, and satisfy $\sum_{j = 1}^\infty \lambda_j < \infty.$  Since $K$ can be associated with a linear integral operator on $L^2(\Istc, \bfp),$ we will refer to $\lambda_j$ as the eigenvalues of $K.$ 

\begin{Theorem}
\label{thm: global}
Suppose assumptions (T1)--(T4) hold.  Then
$$
F\ast_G | X_1,\ldots,X_n \overset{D}{\rightarrow} \sum_{j = 1}^\infty \lambda_j \omega_j \quad \text{almost surely,}
$$
where $\omega_j$ are i.i.d.\ $\chi^2_p$ random variables and $\lambda_j$ are the eigenvalues in \eqref{eq: Keigen}.
\end{Theorem}
While this limiting distribution may appear surprising at first sight given the non-Euclidean setting, they are a result of the fact that central limit theorems can still be derived for data on manifolds (e.g., \ci{bard:13}), including the space of distributions under the Wasserstein metric \cp{pana:16}.  

The limiting distribution obtained in Theorem~\ref{thm: global} depends on unknown parameters, namely the eigenvalues $\lambda_j,$ that must be approximated to formulate a rejection region.  A natural approach would be to estimate the kernel $K$ in \eqref{eq: Keigen} directly, followed by applying a modified Mercer decomposition incorporating the estimated marginal Wasserstein mean.  Thankfully, this complicated approach is not necessary.  Define the quantile conditional covariance kernel $C_Q(s, t) = E\left\{\Cov(Q(s), Q(t) | X)\right\}.$ A simple calculation reveals that $C_Q(s,t) = K(\bQp(s),\bQp(t)),$ so that $\lambda_j$ are also the eigenvalues of $C_Q,$ which can be estimated by
\begin{equation}
\label{eq: CqEst}
\hat{C}_Q(s, t) = \frac{1}{n} \son (Q_i(s) - \hQi(s))(Q_i(t) - \hQi(t)),
\end{equation}
where $\hQi$ are the fitted quantile functions corresponding to densities $\hfi.$ Let $\hat{\lambda}_j$ be the corresponding eigenvalues of $\hat{C}_Q.$  

One can consistently approximate the conditional null distribution of $F\ast_G$ as follows.  For $\alpha \in (0,1)$ and eigenvalue estimates $\hat{\lambda}_j$, let $\hat{b}_\alpha^G$ be the  $1-\alpha$ quantile of $\sum_{j = 1}^{\infty} \hlambda_j\omega_j$, where $\omega_j$ are as in Theorem~\ref{thm: global}.   Computation of this critical value is outlined in Algorithm~\ref{alg: crit} of the Appendix.   The following result on the conditional power $\beta_n^G = \PX(F_G\ast > \hat{b}_\alpha^G)$ follows from Theorem~\ref{thm: global}.  Let $\mk{G}$ denote the collection of distributions on $\R^p \times \mc{D}$ such that model \eqref{eq: Wmodel} holds.

\begin{Corollary}
\label{cor: Gsize}
If $\mc{G} \in \mk{G}$ satisfies $\mc{H}_0^G$ and (T1)--(T4) hold, then $\beta_n^G \rightarrow \alpha$ almost surely.
\end{Corollary}

In addition to having the correct asymptotic size for any null model, we demonstrate the power performance of the above test under a sequence of contiguous alternatives.  To do so, consider a subclass $\mk{G}\ast \subset \mk{G}$ of Wasserstein regression models, with the marginal distribution of $X$ being fixed and satisfying $E(\normE{X}^4)< \infty,$ for which the criteria below are satisfied.  Recall that $\gamma(u) = \Cov(X, R(u)),$ where $R = T\circ S.$
\begin{enumerate}
\item[(G1)] For all $\mc{G} \in \mk{G}\ast,$ (T1) and (T2) hold, with $\sup_{u \in \R} C_T(u,u)$ and $\sup_{u \in \R} C_T^{(1,1)}(u,u)$ uniformly bounded in $\mk{G}\ast.$ In addition, $\sup_{u \in \Ist} \normE{\gamma'(u)}$ and $\sup_{u \in \Ist} \normE{\gamma''(u)}$ are both uniformly bounded in $\mk{G}\ast.$
\item[(G2)] With $L$ being the random Lipschitz constant in (T2), the following uniform moment conditions are satisfied.
\[
\begin{split}
&\limsup_{M \rightarrow \infty} \sup_{\mc{G} \in \mk{G}\ast} E\left[\normE{\tilde{X}}^2\sup_{u \in \R}|T'(u)|\,\, \mathbf{1}\left(\normE{X}^2\sup_{u \in \R} |T'(u)|>M\right)\right] = 0.\\
&\limsup_{M \rightarrow \infty} \sup_{\mc{G} \in \mk{G}\ast} E\left[\normE{\tilde{X}}^3L \,\, \mathbf{1}\left(\normE{\tilde{X}}^3L > M\right)\right] = 0 
\end{split}
\]
\item[(G3)]  There exists $M_1 >1 $ such that
$
\inf_{\mc{G} \in \mk{G}\ast} \inf_{u \in \Ist} \bfp(u) \geq M_1\inv. 
$
\item[(G4)]  There exists $M_2 < \infty$ such that
$
\sup_{\mc{G} \in \mk{G}\ast} \sup_{x \in \mk{S}_X} \sup_{u \in \I_x} \fp(x,u) \leq M_2.
$
\end{enumerate}
\begin{Theorem}
\label{thm: Gpower}
Let $\mk{G}\ast$ satisfy (G1)--(G4), and $a_n$ be a sequence such that $a_n \rightarrow 0$ and $\sqrt{n}a_n \rightarrow \infty.$  Consider a sequence of alternative global hypotheses
$$
\mc{H}_{A,n}^G: \mc{G} \in \mk{G}_{A,n}, \,\,\mk{G}_{A,n} = \left\{\mc{G} \in \mk{G}\ast: E[\dw^2(\fp(X), \bfp)] \geq a_n^2\right\}.
$$
Then the worst case power converges strongly and uniformly to 1, that is, for any $\epsilon > 0,$
$$
\inf_{\mc{G} \in \mk{G}_{A,n}} P\left(\inf_{m \geq n} \beta_m^G \geq 1-\epsilon \right) \rightarrow 1.
$$
\end{Theorem}

\subsection{Test of Partial Effects}
\label{ss: partial}

Beyond a test for no effect, it is often necessary to test the effect for just a single predictor or a subset of them. With $X = (Y,Z)$ as in Section~\ref{ss: est}, under the partial null hypothesis \mbox{$\mc{H}_0^P: \fp(x) = \fzp(y),$}
$$
E\left(\dw^2(\fp(X), \bfp)\right) = E\left(\dw^2(\fzp(Y), \bfp)\right),
$$ 
motivating the partial Wasserstein $F$-statistic
\begin{equation}
\label{eq: partialF}
F\ast_P = \sum_{i = 1}^n\left[ \dw^2(\hfi, \hbfp) - \dw^2(\hfzi, \hbfp)\right], \quad \hfzi = \hfzp(Y_i),
\end{equation}
corresponding to the numerator of the partial $F$-statistic in the multiple linear regression setting.  

Setting $J\T = [-\Sigzy \Sigyy\inv\,\, I_{p-1}]$ and $\Sigma_{Z|Y} = \Sigzz - \Sigzy\Sigyy\inv\Sigyz,$ define the covariance matrix kernel
\begin{equation}
\label{eq: Kseigen}
\begin{split}
K\ast(u,v) &= \Sigma_{Z|Y}^{-1/2}J\T E\left[(X-\mu)(X-\mu)\T C_T(S(u), S(v))\right]J\Sigma_{Z|Y}^{-1/2} \\
&= \sum_{l = 1}^\infty \tau_l \varphi_l(u)\varphi_l\T(v).
\end{split}
\end{equation}
Here, the functions $\varphi_l$ form an orthonormal set in $\left(L^2(\Istc, \bfp)\right)^{p-q}$ and the positive, nonincreasing eigenvalues $\tau_l$ of $K\ast$ satisfy $\sum_{l = 1}^\infty \tau_l < \infty.$ 
\begin{Theorem}
\label{thm: partial}
Under (T1)--(T4), 
$$
F\ast_P | X_1,\ldots,X_n \overset{D}{\rightarrow} \sum_{l = 1}^\infty \tau_l \xi_l^2 \quad \text{almost surely,}
$$
where $\xi_l$ are independent standard normal random variables.  If, in addition, $E(Z|Y)$ is linear in $Y$ and $\Var(Z|Y) = \Sigma_{Z|Y}$ almost surely, then 
$$
\{\tau_l\}_{l = 1}^\infty = \underbrace{\lambda_1,\ldots,\lambda_1}_{\text{$p-q$}}, \underbrace{\lambda_2,\ldots,\lambda_2}_{\text{$p-q$}},\ldots
$$
so that $\sum_{l = 1}^\infty \tau_l \xi_l^2 \overset{D}{=} \sum_{j = 1}^\infty \lambda_j \omega'_j,$ where $\omega'_j$ are i.i.d.\ $\chi^2_{p-q}$ random variables.
\end{Theorem}

Similar to the global case, the $\tau_l$ in \eqref{eq: Kseigen} also correspond to the eigenvalues of the kernel $C_Q\ast(s,t) = K\ast(\bQp(s),\bQp(t)).$ Setting $X_{i,c} = X_i - \xbar,$ a natural estimator is
\begin{equation}
\label{eq: CqstarEst} 
\hat{C}\ast_Q(s,t) = \hSigma_{Z|Y}^{-1/2}\hat{J}\T\left\{\frac{1}{n}\son X_{i,c}X_{i,c}\T (Q_i(s) - \hQi(s))(Q_i(t) - \hQi(t))\right\}\hat{J}\hSigma_{Z|Y}^{-1/2},
\end{equation}
where $\hSigma_{Z|Y}$ and $\hat{J}$ are plug-in estimates.  Let $\hat{\tau}_j$ be the corresponding eigenvalue estimates from $\hat{C}\ast_Q(s,t),$
and $\hat{b}_\alpha^P$ be the $1-\alpha$ percentile of $\sum_{l= 1}^{\infty} \hat{\tau}_l \xi_l^2,$ with $\xi_l$ as in Theorem~\ref{thm: partial}.  If the additional assumptions in the second part of Thereom~\ref{thm: partial} hold, then let $\hat{b}_\alpha^P$ be the $1-\alpha$ quantile of $\sum_{j = 1}^{\infty} \hlambda_j\omega'_j$.  Computation of these critical values can be done in the same was as the global case.
We have the following size and power results for the partial Wasserstein $F$-test, with $\beta_{n}^P = \PX(F_p\ast > \hat{b}_\alpha^P)$ denoting the conditional power as a function of the underlying model $\mc{G}.$  
\begin{Corollary}
\label{cor: Psize}
If $\mc{G} \in \mk{G}$ satisfies $\mc{H}_0^P$ and (T1)--(T4) hold, then $\beta_n^P \rightarrow \alpha$ almost surely.
\end{Corollary}

\begin{Theorem}
\label{thm: Ppower}
Let $\mk{G}\ast$ satisfy (G1)--(G4), and $a_n$ be as in Theorem~\ref{thm: Gpower}.  Consider a sequence of alternative partial hypotheses
$$
\mc{H}_{A,n}^P: \mc{G} \in \mk{G}_{A,n}', \,\,\mk{G}_{A,n}' = \left\{\mc{G} \in \mk{G}\ast: E[\dw^2(\fp(X), \fzp(Y))] \geq a_n^2\right\}.
$$
Then the worst case power converges strongly and uniformly to 1, that is, for any $\epsilon > 0,$
$$
\inf_{\mc{G} \in \mk{G}'_{A,n}} P\left(\inf_{m \geq n} \beta_m^P \geq 1-\epsilon \right) \rightarrow 1.
$$
\end{Theorem}

\subsection{Alternative Testing Approximations}
\label{ss: testAlt}

As an alternative to estimating the eigenvalues in the limiting distributions of $F\ast_G$ and $F\ast_P,$ Satterthwaite's approximation \cp{satt:41,shen:04} can elso be employed.  Using the global test as an example, we approximate the null distribution of the test statistic by $a\chi^2_m$, where $a, m > 0$ are scalars chosen to satisfy the moment matching conditions
$
am = p\sum_{j = 1}^\infty \lambda_j,$ $a^2m = p \sum_{j = 1}^\infty \lambda_j^2.
$
Using this approximation, one does not need to estimate the individual eigenvalues $\lambda_j$, given the equalities \cp{hsin:15}
$$
\sum_{j = 1}^\infty \lambda_j = \int_0^1 C_Q(t, t)\d t, \quad \sum_{j = 1}^\infty \lambda_j^2 = \int_0^1\int_0^1 C_Q^2(s, t)\d s \d t.
$$
Hence, one can compute the approximate values $a$ and $m$ using the corresponding norms of the estimate $\hat{C}_Q.$  This alternative approach is outlined in Algorithm~\ref{alg: critS} of the Appendix.

Finally, as the limiting distributions in Theorems~\ref{thm: global} and \ref{thm: partial} are the result of underlying central limit theorems, one may employ a bootstrap approach to testing these hypotheses.  Since the inference is conditional on the observed predictors $X_i,$ a natural approach is to perform a residual transport bootstrap.  Using the global test as an example, let $\hat{T}_{i,0} = Q_i \circ \hFp\ast$ be the approximate versions of the residual transports $T_i$ under $\mc{H}_0^G.$  Obtain $B$ independent bootstrap samples $\{\tilde{T}_{i}^b\}_{i = 1}^n,$ $b = 1,\ldots B$, by sampling with replacement from the $\hat{T}_{i,0},$ and form the bootstrapped quantile functions $\tilde{Q}_{i}^b = \tilde{T}_i^b \circ \hat{Q}_\oplus\ast.$  Then, compute bootstrap estimates $\hfp^b(x)$ and $\hat{f}_\oplus^{*,b}$ using the data $(X_i,\tilde{Q}_i^b),$ $i = 1,\ldots,n.$  Finally, compute the bootstrap statistics $ \tilde{F}^{*,b}_G = \son \dw^2(\hfp^b(X_i),\hat{f}_\oplus^{*,b}).$  The bootstrap $p$-value then becomes $$ \frac{\# \left\{\tilde{F}^{*,b}_G > F_G\ast\right\} + 1}{B+1}.$$  This global residual bootstrap test is outlined in Algorithm~\ref{alg: critB} of the Appendix.  For the partial test, this residual bootstrap can only be employed if the support of the densities is fixed, since otherwise the null residual transports $\hat{T}_{i,0} = Q_i \circ \hFzp(Y_i)$ will have different supports.

\section{Confidence Bands}
\label{sec: cb}

We now develop methodology for producing a confidence set for $\fp(x),$ where $x$ is considered to be fixed.  In similar settings where one desires to make a confidence statement for a functional parameter, such as nonparametric regression \cp{euba:93:1,clae:03}, mean and covariance estimation in functional data analysis \cp{degr:11,wang:09:2,cao:12}, and the varying coefficient model \cp{fan:08}, one can either build pointwise or simultaneous bands.  In the current setting of Wasserstein regression for density response data, the constraints inherent to the density targets $\fp(x)$ render pointwise confidence bands of little use.

Hence, the approach we take will be motivated by simultaneous confidence bands.  Here, the descriptor ``simultaneous" refers to the argument $u$ of the functional parameter $\fp(x,u),$ and not to the specific regressor value $x$ under consideration.
Let $g$ be a generic functional parameter of interest, where we assume that $g$ is bounded.  Given an estimator $\hat{g},$ a typical approach to formulating a simultaneous confidence band is to show that $b_n(\hat{g}(u) - g(u))/a(u)$ converges weakly to a limiting process (usually Gaussian) in the space of bounded functions under the uniform metric \cp{well:96}, where $a>0$ is a scaling function and $b_n\inv$ is the rate of convergence.  By an application of the continuous mapping theorem, one can then obtain a confidence band for $g$ of the form
$$
\{g\ast: \hat{g}(u) - ca(u)b_n\inv \leq g\ast(u) \leq \hat{g}(u) + ca(u)b_n\inv\,\, \textrm{for all }u\}.
$$
This band corresponds to all functions that are almost everywhere between the lower and upper bounds, and is the closest one can get to a confidence interval in function space.  This partial ordering is induced explicitly by the uniform metric, and indicates why simultaneous confidence bands are so useful for functional parameters, in that one can visualize the entire set graphically.  We will explore two different approaches to formulating simultaneous confidence bands.  The first arises naturally from the Wasserstein geometry, and provides a distributional band for either $\Qp(x)$ or $\Fp(x)$, but not the density parameter.  In the second approach, we strengthen the convergence results and utilize the delta method to construct a simultaneous confidence band for $\fp(x).$

\subsection{Intrinsic Wasserstein-$\infty$ Bands}
\label{ss: IWB}
 
The first method is directly related to the geometry imposed by the Wasserstein-$\infty$ metric; see \eqref{eq: wassInf}. Specifically, if $\hat{V}_x = \hQp(x) \circ \Fp(x)$ is the optimal transport from the target to the estimate, then
$$
d_\infty(\fp(x), \hfp(x)) = \fp(x)\textrm{-}\esssup_u |\hat{V}_x(u) - u| = \sup_{u \in \I_x}|\hat{V}_x(u) - u|.
$$
It is then natural to establish weak convergence (denoted by $\rsa$) of the process $\hat{V}_x(u)$ within the space of bounded functions on $\I_x,$ denoted $L^\infty(\I_x).$  Define $\Lambda = E(\tilde{X}\tilde{X}\T)$ and the covariance kernel
\begin{equation}
\label{eq: tK}
\tilde{K}(u,v) = \Lambda\inv E\left[\tilde{X}\tilde{X}\T C_T(S(u), S(v))\right] \Lambda\inv, \quad u,v \in \Istc.
\end{equation}
\begin{Theorem}
\label{thm: fittedVal}
Suppose that (T1)--(T4) hold.  Then there exists a zero-mean Gaussian process $\mc{M}_x$ on $L^\infty(\I_x)$ such that
$$
\sqrt{n}(\hat{V}_x - \id) | X_1,\ldots,X_n \rsa \mc{M}_x \quad \text{almost surely.}
$$
With $u_x = \bQp\circ \Fp(x,u)$ for any $u \in \I_x,$ the covariance of $\mc{M}_x$ is
\begin{equation}
\label{eq: CxCov}
\mc{C}_x(u,v)  = \tilde{x}\T \tilde{K}\left(u_x, v_x \right) \tilde{x}.
\end{equation}
\end{Theorem}

As in the testing procedures, estimation of the covariance kernel $\tilde{K}$ is simplified by moving to quantile functions.  Set $D_Q(s,t) = \tilde{K}(\bQp(s),\bQp(t)),$ with estimate
\begin{equation}
\label{eq: DqEst}
\hat{D}_Q(s,t) = \hLambda\inv\left\{\frac{1}{n}\son \tilde{X}_i\tilde{X}_i\T(Q_i(s) - \hat{Q}_i(s))(Q_i(t) - \hat{Q}_i(t))\right\}\hLambda\inv, 
\end{equation}
where $\hLambda = n\inv \son \tilde{X}_i \tilde{X}_i\T.$ This leads to
\begin{equation}
\label{eq: CxEst}
\hat{\mc{C}}_x(u,v) = \tilde{x}\T \hat{D}_Q(\hFp(x,u), \hFp(x,v)) \tilde{x}, \quad u,v\in [\hQp(x,0), \hQp(x,1)].
\end{equation}

Let $m_\alpha$ be the $1-\alpha$ quantile of the distribution of
\[
\zeta_x:=\sup_{u \in \I_x} \mc{C}_x(u,u)^{-1/2}\left|\mc{M}_x(u)\right|.
\]
Since $m_\alpha$ is unknown, we estimate it as follows. Observe that $\mc{N}_x = \mc{M}_x \circ \Qp(x)$ is a Gaussian process on $L^\infty\zo$ with covariance $\tilde{x}\T D_Q(s,t)\tilde{x}.$  Conditional on the data, let $\hat{\mc{N}}_x$ be a zero-mean Gaussian process with covariance $\tilde{x}\T \hat{D}_Q(s,t)\tilde{x}$. Define
$$
\hat{\zeta}_x = \sup_{t \in \zo} \left[\tilde{x}\T \hat{D}_Q(t,t)\tilde{x}\right]^{-1/2}\left|\hat{\mc{N}}_x(t)\right|,
$$
and set $\hat{m}_\alpha$ as the $1-\alpha$ quantile of $\hat{\zeta}_x.$ Then the Wasserstein-$\infty$ confidence band for $\fp(x)$ is
\begin{equation}
\label{eq: WinfCB}
C_{\alpha,n}(x) = \left\{g \in \mc{D}: \sup_{u \in \hat{\I}_x} \frac{|T_{g,x}(u) - u|}{\hat{\mc{C}}_x(u, u)^{1/2}} \leq \frac{\hat{m}_\alpha}{\sqrt{n}} \right\}, \quad T_{g,x} = G\inv \circ \hFp(x).
\end{equation}
We have the following corollary of Theorem~\ref{thm: fittedVal}.

\begin{Corollary}
\label{cor: CB}
Suppose (T1)--(T4) hold.  If $\inf_{u \in \Ist} C_T(u,u) > 0,$ then
$$
\PX\left(\fp(x) \in C_{\alpha, n}(x)\right) \rightarrow 1- \alpha \quad \text{almost surely.}
$$
\end{Corollary}
An important case arising in practice that is ruled out by the requirement of strictly positive covariance is when the support of $\mk{F}$ is some fixed interval $\I,$ so that the random transport $T$ is necessarily fixed at the boundaries.  In this case, a slight adjustment can be made as outlined in Corollary~\ref{cor: CB_fixed} of the Appendix.

Next, we demonstrate a connection between these simultaneous Wasserstein-$\infty$ bands and the usual partial stochastic ordering of distributions.  
Recall that a cdf $F$ is said to be stochastically greater than another $G$, written $F \succ G$, if $F(u) \leq G(u)$ for all $u.$  For $F_2 \succ F_1,$ define the bracket
$$
\left[ F_1, F_2\right] = \left\{g \in \mc{D} : F_2 \succ G \succ F_1\right\}
$$
consists of all distributions in $\mc{D}$ which lie between $F_1$ and $F_2$ in the stochastic ordering.  From \eqref{eq: WinfCB}, we deduce that the simultaneous confidence band consists of all densities $g \in \mc{D}$ such that
$$
M_{L}(u) \leq G\inv \circ \hFp(x, u) \leq M_U(u), \quad (M_L(u), M_U(u)) = u \pm n^{-1/2}\hat{m}_\alpha\hat{\mc{C}}_x(u,u)^{1/2}.
$$
Define $T_L$ to be the unique projection (in $L^2[0,1]$) of $M_L$ onto the closed and convex set of non-decreasing functions $M$ for which $M(u) \geq M_L(u).$  Similarly, $T_U$ is unique largest nondecreasing function below $M_U.$  Note that $T_L(u) \leq T_{g,x}(u) \leq T_U(u)$ for all $g \in C_{\alpha,n}(x).$  Then the cdf bounds $F_L = \hFp(x) \circ T_L\inv$ and $F_U = \hFp(x)\circ T_U\inv$ represent the Wasserstein-$\infty$ band, i.e.\ $C_{\alpha,n}(x) = \left[F_L, F_U\right].$  Algorithm~\ref{alg: QBand} in the Appendix outlines the steps for computing  $C_{\alpha,n}(x)$ in practice.

\subsection{Wasserstein Density Bands}
\label{ss: WDB}

One drawback of the above simultaneous confidence bands is that they do not readily translate to the space of densities, since, if $F_1 \prec F_2$ in the stochastic ordering, their derivatives need not satisfy $f_1\geq f_2.$  
Thus, we take a second approach to form a confidence band in density space, based on the direct difference $\hfp(x) - \fp(x)$ rather than the optimal transport map between these distributions. 

An interesting challenge associated with this approach is that the supports of the target $\fp(x)$ and its estimate $\hfp(x)$ may differ, so that convergence may be ill-behaved near the boundaries.  To resolve this, for any $\delta \in (0,1/2),$ define $\I_x^\delta = [\Qp(x,\delta), \Qp(x, 1-\delta)].$  We consider conditional weak convergence of the process $\sqrt{n}(\hfp(x) - \fp(x))$ in the space $L^\infty(\I_x^\delta).$  For $l,m \in \{0,1\}$ and $\tilde{K}$ as in \eqref{eq: tK}, define $\tilde{K}^{(l,m)} = \frac{\partial^{l+m}}{\partial u^l \partial v^m} \tilde{K}.$ Lastly, set
$$
\mathbb{K}(u,v) = \begin{pmatrix}
 \tilde{K}(u,v) &\left[\bfp(v)\right]\inv \tilde{K}^{(0,1)}(u, v) \\ 
\left[\bfp(u)\right]\inv \tilde{K}^{(1,0)}(u, v) & \left[\bfp(u)\bfp(v)\right]\inv \tilde{K}^{(1,1)}(u, v)
\end{pmatrix}, \,\, u,v \in \Ist.
$$

\begin{Theorem}
\label{thm: densCLT}
Suppose (T1)--(T4) hold, and that ${\bfp}$ is continuously differentiable.  Then, for almost all $x,$ there exists a zero-mean Gaussian process $\mc{F}_x$ on $L^\infty(\I_x^\delta)$ such that
$$
\sqrt{n}\left(\hfp(x) - \fp(x)\right) | X_1,\ldots,X_n \rsa \mc{F}_x \quad \text{almost surely.}
$$
With $u_x = \bQp \circ \Fp(x,u)$ and $c_u\T = (\partial / \partial u \fp(x,u), -\fp^2(x,u))$  for any $u \in \I_x^\delta$, and $\otimes$ denoting the Kronecker product, the covariance of $\mc{F}_x$ is
\begin{equation}
\label{eq: FxCov}
\mc{R}_x(u,v) = c_u\T \left(\tilde{x}\T \otimes \mathbb{K}(u_x,v_x) \otimes \tilde{x}\right) c_v.
\end{equation}
\end{Theorem}

We now describe the method for estimating $\mc{R}_x.$  With $\hat{D}_Q$ as in \eqref{eq: DqEst}, let $\hat{D}_Q^{(l,m)} = \frac{\partial^{l+m}}{\partial s^l \partial t^m} \hat{D}_Q.$  Then, define $\mathbb{D}(s,t) = \mathbb{K}(\bQp(s),\bQp(t))$ and its estimate
\begin{equation}
\label{eq: DmatEst}
\hat{\mathbb{D}}(s,t) = \begin{pmatrix}
\hat{D}_Q(s,t) & \hat{D}_Q^{(0,1)}(s,t) \\ \hat{D}_Q^{(1,0)}(s,t) & \hat{D}_Q^{(1,1)}(s,t),
\end{pmatrix} \quad s,t \in [\delta, 1-\delta],
\end{equation}
and set 
$$
\hat{\mathbb{K}}(u,v) = \hat{\mathbb{D}}(\hFp(x,u), \hFp(x,v)), \quad u,v \in \hat{\I}_x^\delta = [\hQp(x,\delta), \hQp(x, 1-\delta)].
$$  
Next, the chain rule implies that $\frac{\partial}{\partial u} \fp(x,u) = -\fp^3(x,u) \frac{\partial}{\partial t} \qp(x, \Fp(x,u)).$  Hence, define $\frac{\partial}{\partial t}\hqp(x,t) = n\inv \son s_{in}(x)q_i'(t)$ for $t \in [\delta, 1-\delta],$ yielding 
$$
\frac{\partial}{\partial u} \hfp(x,u) = -\hfp^3(x,u)\frac{\partial}{\partial t} \hqp(x, \hFp(x,u)), \quad u \in \hat{\I}_x^\delta.
$$  
Finally, for $u,v \in \hat{\I}_x^\delta,$ set 
\begin{equation}
\label{eq: RxEst}
\hat{\mc{R}}_x(u,v) = \hat{c}_u\T\left(\tilde{x}\T \otimes \hat{\mathbb{K}}(u,v) \otimes \tilde{x}\right) \hat{c}_v, \quad \hat{c}_u\T = \left(\frac{\partial}{\partial u} \hfp(x,u), \, -\hfp^2(x,u)\right).
\end{equation}

Let $l_\alpha$ be the unknown $1-\alpha$ quantile of 
\[
\xi_x= \sup_{u \in \I_x^\delta} \mc{R}_x(u,u)^{-1/2} |\mc{F}_x(u)|.
\]
Similar to the case of the Wasserstein-$\infty$ band, conditional on the data, let $\hat{\mc{J}}_x(t)$ be a zero-mean Gaussian process on $(0,1),$ with covariance
$
\hat{\mc{R}}_x(\hQp(x,s), \hQp(x,t)).
$
Define $\hat{l}_\alpha$ as the the $1-\alpha$ quantile of $$\hat{\xi}_x = \sup_{t \in [\delta, 1-\delta]} \left[\hat{\mc{R}}_x(\hQp(x,t), \hQp(x,t))\right]^{-1/2}\left|\hat{\mc{J}}_x(t)\right|.$$
Then the nearly simultaneous Wasserstein density confidence band is
\begin{equation}
\label{eq: densCBd}
B_{\alpha,n}(x) = \left\{g \in \mc{D}:\, \sup_{u \in \hat{\I}_x^\delta} \frac{|g(u) - \hfp(x,u)|}{\hat{\mc{R}}_x(u,u)^{1/2}} \leq \frac{\hat{l}_\alpha}{\sqrt{n}}\right\}.
\end{equation}
See Algorithm~\ref{alg: densBand} in the Appendix for an implementation of this confidence band.

The final corollary demonstrates almost sure convergence of the coverage rate of the Wasserstein density confidence band.  Since the estimation of the covariance $\mc{R}_x$ is considerably more complex than for the Wasserstein-$\infty$ band, additional smoothness assumptions are required on the random transport map $T.$
\begin{Corollary}
\label{cor: WDB}
Suppose the assumptions of Theorem~\ref{thm: densCLT} hold, and that $\inf_{u \in \I_x^\delta} \mc{R}_x(u,u) > 0.$  Furthermore, assume that $T$ is twice differentiable almost surely, with $T''$ having Lipschitz constant $\tilde{L}$ and satisfying 
$$
E\left(\normE{\tilde{X}}^6\sup_{u \in \R} |T''(u)|^2\right)< \infty, \quad E(\tilde{L})< \infty.
$$  
Then
$$
\PX\left(\fp(x) \in B_{\alpha,n} \right) \rightarrow 1-\alpha \quad \textrm{almost surely}.
$$ 
\end{Corollary}

\section{Simulations}
\label{sec: sims}

Extensive simulations were conducted to investigate the empirical performance of the $F$-tests and confidence bands developed in Sections~\ref{sec: tests} and \ref{sec: cb}.  Data were generated according to model \eqref{eq: Wmodel} for increasing sample sizes and different random optimal transport processes $T$.  Furthermore, to assess the robustness of our procedures when densities are only indirectly observed through samples, simulations were also conducted using only raw data generated from the random densities. For brevity, results for the indirectly observed case are reported in the Appendix.

We first describe the simulation settings for the Wasserstein $F$-tests.  Following Example~\ref{exm: LocScale}, set $f_0$ as the standard normal distribution, truncated to the interval $[-2.5, 2.5]$ and renormalized to have integral one.  Bivariate predictors $X_i = (X_{i1},X_{i2})$ were generated as independent uniform random variables on $[-0.5, 0.5].$  Let $\nu(x) = \alpha_1 x_1 + \alpha_2 x_2$ and $\tau(x) = 2 + \beta_1 x_1 + \beta_2 x_2,$ where the parameters $\alpha_j$, $\beta_j$ were varied to increase signal strength for assessing power performance.  The conditional Wasserstein mean density was
\begin{equation}
\label{eq: simMean}
\fp(x,u) = \frac{1}{\tau(x)}f_0\left(\frac{u - \nu(x)}{\tau(x)}\right).
\end{equation}
To produce the random densities $\mk{F}_i,$ random optimal transport maps $T_i$ were generated in two different ways.  In the first setting, $T_i(u) = V_{i1} + V_{i2}u,$ where $V_{i1} \sim \mc{U}(-0.5, 0.5)$ and $V_{i2} \sim \mc{U}(0.5, 1.5)$ were generated independently, yielding linear optimal transport maps.  In the second setting, following the method described in Section 8.1 of \ci{pana:16}, nonlinear optimal transport maps were generated.  Specifically, define template transport maps $M_0(u) = u$ and
$$
M_k(u) =  u - \frac{\sin(ku)}{|k|}, \quad k \neq 0.
$$
For $J = 10,$ $(W_{i1}, \ldots, W_{iJ})$ were sampled from a order $J$ Dirichlet distribution, so that $W_j > 0$ and $\sum_{j = 1}^J W_j = 1.$ Then $K_{ij},$ $j = 1,\ldots,J$, were sampled independently and uniformly from $(-0.250, -0.125,  0,  0.125,  0.250)$.  The resulting optimal transport map in the second setting was $T_i(u) = \sum_{j = 1}^J W_{ij} M_{K_{ij}}(u).$  Finally, the random density was generated as $$\mk{F}_i(u) = \fp\left(X_i, T_i\inv(u)\right) (T_i \inv)'(u).$$  In the case when the densities were not directly observed, secondary samples of size 300 from each $\mk{F}_i$ were generated, and local linear smoothing of the empirical quantile function was used to estimate the quantile functions $\hat{Q}_i$ for use in the testing algorithms.  The Matlab function `smooth' was used with default choice of the smoothing parameter.  
\begin{figure}[t!]
\centering
\subcaptionbox{}[5in]{\includegraphics[scale = 0.9]{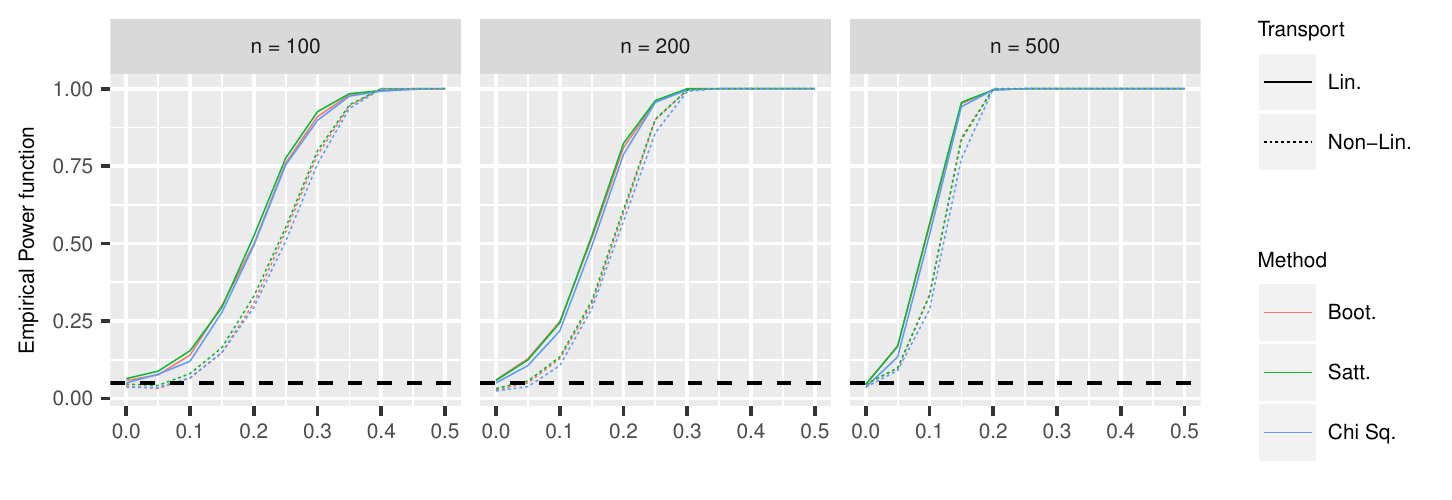}} \\
\subcaptionbox{}[5in]{\includegraphics[scale = 0.9]{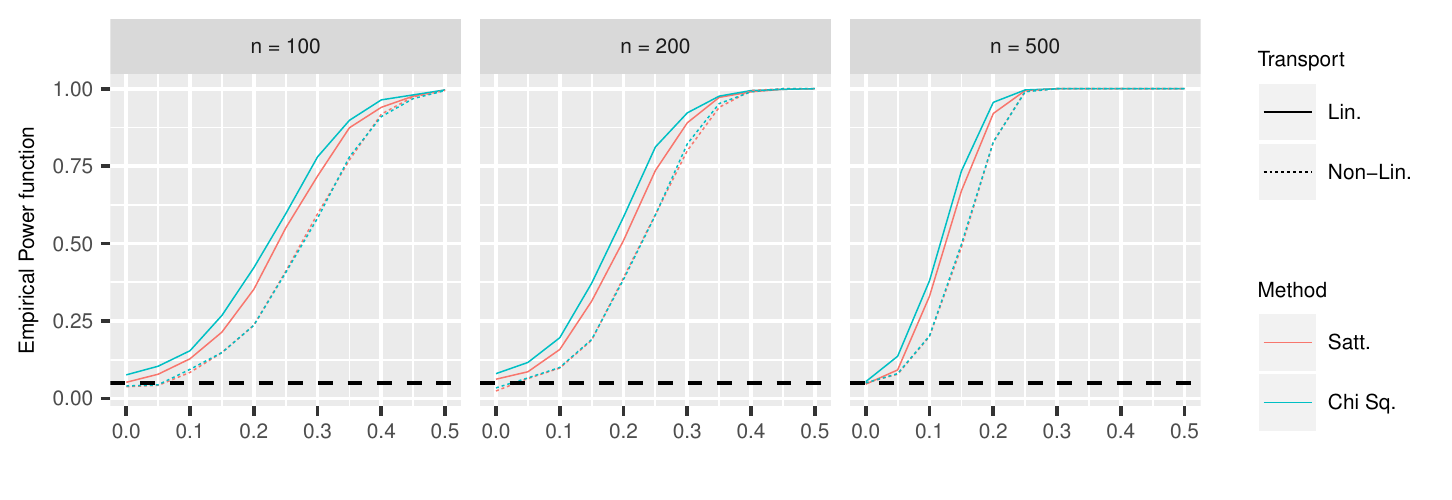}}
\caption{\footnotesize Power curves for global (top) and partial (bottom) Wasserstein $F$-tests, for sample sizes $n = 100, 200, 500.$ The dotted horizontal line represents the nominal level $\alpha = 0.05.$ \label{fig: test}}
\end{figure}

In both the global and partial F tests simulations, the empirical size approached the nominal size $\alpha = 0.05$ as the sample size grew, and the empirical power function increased with the increasing signal strength. To create models satisfying global null and alternative hypotheses, $\alpha_1 = \alpha_2 = \beta _1 = \beta_2$ were set to be equal, with the common value running through $(0, 0.5)$. In the partial $F$-test, $\alpha_1$ and $\beta_1$ were fixed as $\alpha_1 = 2$ and $\beta_1 = 1$ respectively, while $\alpha_2 = \beta_2$ varied across $(0, 0.5)$. Therefore, the null hypothesis in the partial Wasserstein $F$-test was $\mathcal{H}_0^P: f_{\oplus}(x) = f_\oplus^0(x_1)$. 

For each of three sample sizes $n = 100, 200, 500$, five hundred simulations were used to compute empirical power curves shown in Figure~\ref{fig: test}.  For the global tests, performance was similar among the $\chi^2$ mixture and alternative Satterthwaite and bootstrap tests outliend in Section~\ref{ss: testAlt}.  The nominal sizes converged to the right level, and power converged to one with increasing values of the common parameter value.  Power was generally higher for the linear transport map setting, as expected.  The same pattern is observed for the partial test, although the Satterthwaite alternative more accurately maintained the nominal level for low sample sizes.  The bootstrap test could not be implemented for the partial test due to the fact that the support of the density varies with $x.$  Similar results for indirectly observed densities are given in Figure~\ref{fig: test_est} in the Appendix.  For both global and partial tests, the difference in performance is more prominent in low sample sizes (in $n$, the number of densities), with the power converging at a slower rate and slightly larger deviations from the nominal level compared to the case of observed densities.

\begin{table}[t]
    \centering
    \caption{Confidence band error rates for fully observed densities over 500 simulation runs}
    \begin{tabular}{l|ccccccc}
\toprule
& \multicolumn{6}{c}{Wasserstein-$\infty$ band} \\
        &  \multicolumn{3}{c}{linear transport map} & & \multicolumn{3}{c}{nonlinear transport map}  \\\cline{2-4}\cline{6-8}
& n = 100 & n = 200 & n = 500 &  & n = 100  & n = 200 & n = 500 \\ 
\hline 
x = -0.30 & 0.050 & 0.050 & 0.046 & & 0.044 & 0.074 & 0.048 \\ 
x = -0.24 & 0.056 & 0.040 & 0.058 & &  0.044 & 0.068 & 0.044 \\ 
x = -0.18 & 0.058 & 0.042 & 0.058 &  & 0.052 & 0.064 & 0.046 \\ 
x = -0.12 & 0.056 & 0.036 & 0.062 & &  0.058 & 0.064 & 0.054 \\ 
x = -0.06 & 0.056 & 0.040 & 0.048 & &  0.058 & 0.070 & 0.058 \\ 
x = 0 & 0.056 & 0.042 & 0.042 &  & 0.052 & 0.074 & 0.056 \\ 
x = 0.06 & 0.062 & 0.052 & 0.036  & & 0.052 & 0.072 & 0.056 \\ 
x = 0.12 & 0.060 & 0.050 & 0.036 &  & 0.054 & 0.072 & 0.058 \\ 
x = 0.18 & 0.062 & 0.052 & 0.040 &  & 0.054 & 0.066 & 0.052 \\ 
x = 0.24 & 0.068 & 0.052 & 0.042 & &  0.058 & 0.064 & 0.058 \\ 
x = 0.30 & 0.076 & 0.044 & 0.050 & &  0.058 & 0.062 & 0.056 \\ 
\hline
& \multicolumn{6}{c}{Wasserstein density band*} \\
        &  \multicolumn{3}{c}{linear transport map} & & \multicolumn{3}{c}{nonlinear transport map}  \\\cline{2-4}\cline{6-8}
& n = 100 & n = 200 & n = 500 &  & n = 100  & n = 200 & n = 500 \\ 
\hline
x = -0.30 & 0.062 & 0.056 & 0.054 & & 0.066 & 0.074 & 0.064 \\ 
x = -0.24 & 0.056 & 0.056 & 0.052 & &  0.062 & 0.072 & 0.058 \\ 
x = -0.18 & 0.052 & 0.052 & 0.054 & &  0.054 & 0.058 & 0.050 \\ 
x = -0.12 & 0.034 & 0.042 & 0.052 &  & 0.048 & 0.052 & 0.046 \\ 
x = -0.06 & 0.038 & 0.038 & 0.038 &  & 0.046 & 0.050 & 0.044 \\ 
x = 0 & 0.030 & 0.026 & 0.040 & &  0.058 & 0.050 & 0.052 \\ 
x = 0.06 & 0.028 & 0.030 & 0.036 &  & 0.066 & 0.054 & 0.058 \\ 
x = 0.12 & 0.048 & 0.046 & 0.032 &  & 0.068 & 0.046 & 0.062 \\ 
x = 0.18 & 0.052 & 0.048 & 0.034 &  & 0.074 & 0.060 & 0.062 \\ 
x = 0.24 & 0.062 & 0.046 & 0.036 & &  0.074 & 0.066 & 0.062 \\ 
x = 0.30 & 0.064 & 0.048 & 0.036 & &  0.086 & 0.076 & 0.064 \\
 \bottomrule  
\end{tabular}
\begin{tablenotes}
\centering
      \item *$\delta = 0.1$ used to avoid boundary issue
    \end{tablenotes}
    \label{tab: cb}
\end{table}

For simplicity, in the simulations for confidence bands, a single predictor $X_i \sim \mc{U}(-0.5, 0.5)$ was used.  With $\nu(x) = 2x$ and $\tau(x) = 2 + x,$ the mean was again given by \eqref{eq: simMean}, where $f_0$ was the same as in the testing simulations.  The random optimal transports and case of indirectly observed densities were handled the same as in the testing simulations as well. Both types of confidence intervals were computed for predictor values $x \in ( -0.30, -0.24, \ldots, 0.24,  0.30)$, with $\delta = 0.1$ being used for the density bands.  In addition, for the case of indirectly observed densities, coverage of Wasserstein-$\infty$ bands was also measured using $\delta = 0.1$ due to boundary effects associated with the preliminary smoothing, as the random quantile functions are very steep near the boundary.  
The error rates in Table \ref{tab: cb} were computed using 500 runs for each setting.  Overall, error rates improved as sample size increased, with error rates near $\alpha = 0.05$ at all $x$ values when $n = 500.$

Table~\ref{tab: cb_est} in the Appendix contains the corresponding results when for indirectly observed densities.  Note that Algorithm~\ref{alg: densBand} for computing the Wasserstein density band also requires estimation of $q_i$ and $q_i'$ in addition to the quantile function $Q_i.$  In our experiments, $\hat{q}_i$ and $\hat{q}_i'$ were computed by numerical differentiation of $\hat{Q}_i.$  Clearly there are alternative ways in which one might estimate these functions, but numerical differentiation was chosen because it represents a worst-case scenario in order to reveal sensitivities of the density bands to errors induced by this preprocessing step.  Convergence to the nominal coverage rate was slower for indirectly observed densities.  In particular, the Wasserstein-$\infty$ band suffered from the aforementioned boundary issues associated with local linear estimation of quantile functions that are steep near the boundary.  However, with $n = 500,$ all but one of the error rates were below 0.1, compared to the nominal rate $\alpha = 0.05.$

\section{Application to Stroke Data}
\label{sec: stroke}

Intracerebral hemorrhage (ICH), caused by small blood vessel ruptures inside the brain, is the second most common stroke subtype \cp{morg:10}. Computed tomography (CT) is the most utilized imaging modality to diagnose and study ICH in clinical settings. Various studies of ICH have revealed the importance of the density of the hematoma, in addition to important factors such as the size and location of the hematoma inside brain parenchyma.  \cp{barr:09,delc:16,sala:19}. The density of the hematoma is not homogenous, and common practice is to summarize this important feature by some set of subjective scalar measures that are often obtained by visual inspection and are thus user-dependent. For example, \ci{boul:16} demonstrated the importance of the CT hypodensity, a binary variable indicating the presence of low-density regions within the hematoma. Instead of using any particular summary, one can study the distribution of hematoma density (i.e., a probability density function of hematoma density) throughout the entire hematoma as a functional object. Specifically, the hematoma density is measured on the Hounsfield scale (0 - 100 HU) for each voxel within the hematoma, and these can be used to produce a probability density function in a variety of ways.  In this application, the hematoma densities were first combined into a histogram, followed by smoothing to obtain a probability density function; see Fig.~\ref{fig: ObsFitted}.  Other methods, such as local linear smoothing to obtain quantile function estimates as illustrated in the simulations, could also be used.

To illustrate the utility of the proposed inferential methods for Wasserstein regression, we considered a study of $n = 393$ ICH anonymized subjects and analyzed the associations between 5 clinical and 4 radiological variables as predictors of the head CT hematoma densities as distributional responses. The clinical predictors were age, weight, history of diabetes, and two variables indicating history of coagulopathy (Warfarin and AnitPt).  Radiological predictors included the logarithm of hematoma volume, a continuous index of hematoma shape (Shape), presence of a shift in the midline of the brain, and length of the interval between stroke event and the CT scan (TimetoCT).  These predictors were selected based on the natural history and published clinical studies \cp{salm:09,al:18}.  A complete description of the data source can be found in \ci{heve:18}. As seen in Figure~\ref{fig: ObsFitted}, important features that vary from subject to subject are the location and mode of the hematoma density distribution, as well as the spread, where some hematoma densities are homogeneous and concentrated around the single mode, while others are more heterogeneous or even bimodal. Thus, while the mean is an important aspect of the hematoma density, this natural summary fails to capture other forms of variability that are automatically taken into account by the Wasserstein regression model. 

To assess the goodness of fit, the Wasserstein coefficient of determination \cp{mull:19:3} was computed as
$$
\hat{R}_{\oplus}^2 = 1 - \frac{\son \dw^2(\mk{F}_i, \hfi)}{\son \dw^2(\mk{F}_i, \hbfp)} = 0.2242,
$$
representing the fraction of Wasserstein variability explained by the model. As a comparison, a standard linear functional response model \cp{fara:97} was also fit, after transforming the densities into a linear function space using the log quantile density (LQD) transformation of \ci{mull:16:1}.  Because the hematoma densities fail to have common supports, as required by the LQD transformation, the densities were regularized to be strictly positive on $(0,100).$  This was done by appropriately mixing each density with the uniform distribution on $(0,100)$, with the mixture coefficient chosen so that the densities were all bounded below by $0.0005.$  This alternative LQD model yielded fitted densities by means of the inverse LQD transformation.  The LQD method suffered from quite poor recovery of the quantile functions near the boundaries $t \in \{0,1\},$ due to the fact that the observed densities decay to zero at the boundary of their supports.  The Wasserstein distance in \eqref{eq: wass2} was sensitive to these errors, and resulted in a much lower Wasserstein coefficient of determination of $\hat{R}_\oplus^2 = 0.0389,$  so that the Wasserstein regression model provided a substantial improvement on the LQD approach.  

\begin{table}[t]
\caption{LQD transformation method and Wasserstein regression $p$-values 
 \label{tab: FullModel}}
\begin{tabular}{l | l | r r}
Category & Predictor& Wasserstein  & LQD \\ \hline
\multirow{5}{*}{Clinical} & Age & 0.862 &  0.761 \\ 
 & {Weight} & $<0.001$ &  0.479  \\ 
 & {DM} & 0.034 &  0.742  \\ 
  & Warfarin  & 0.298 & 0.640 \\ 
  & {AntiPt} & 0.078 &  0.900   \\ \hline \hline
 \multirow{4}{*}{Radiological} &  $\log$({Volume}) & $ <0.001$ &  $ <0.001$ \\
&  {Hematoma Shape} & $<0.001$ &  $<0.001$ \\
&  {Midline Shift} & $<0.001$ & 0.039 \\ 
&  TimetoCT &  0.902 & 0.657\\ 
   \hline
\end{tabular}
\end{table}

The $p$-value for the global Wasserstein $F$-test was approximately zero using the mixture $\chi^2$ test, giving strong evidence of a significant regression relationship between the hematoma densities and candidate predictors. The alternative testing procedures of Section~\ref{ss: testAlt} gave similar results. Significance of the effects of individual predictors are found in Table~\ref{tab: FullModel}, where these were computed via partial Wasserstein $F$-tests.  As a baseline for comparison, the $F$-tests of \ci{shen:04} were performed on the LQD functional linear model.  The global $F$-test on the LQD model was also approximately zero, with partial $p$-values given in the last column of Table~\ref{tab: FullModel}.  Both models identified hematoma size and shape, as well as presenece of midline shift, as important factors in determining the distribution of density within the hematoma.  However, the Wasserstein regression model identifies two clinical variables related to weight and diabetes as also affecting the hematoma density distribution.

\begin{figure}[t]
\centering
\subcaptionbox{Wasserstein-$\infty$ band \label{fig: CB_CDF}}[2.4in]{\includegraphics[width=2.4in, height=2.4in]{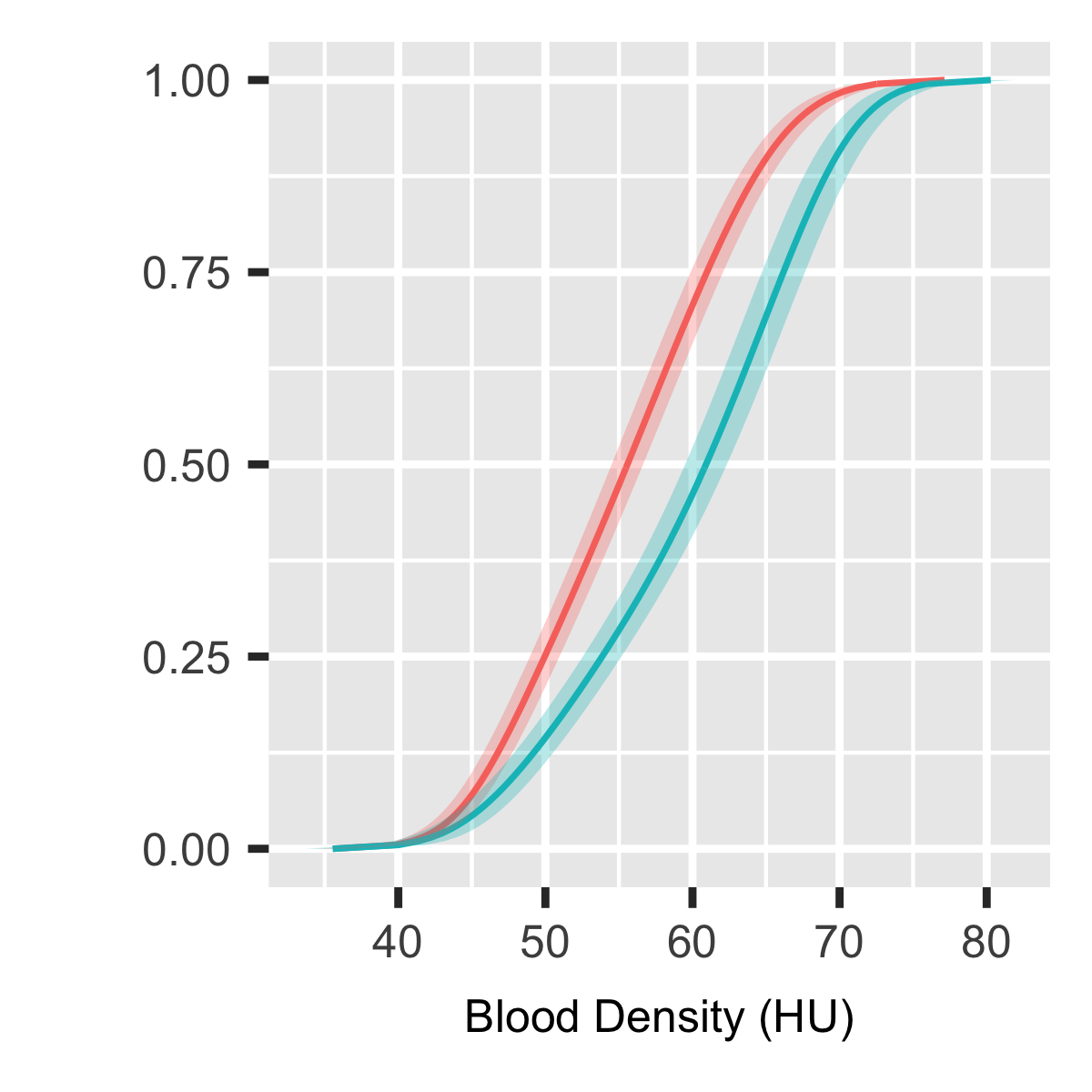}} \hspace{-1cm}
\subcaptionbox{Wasserstein density band \label{fig: CB_DENS}}[2.88in]{\includegraphics[width=2.88in, height=2.4in]{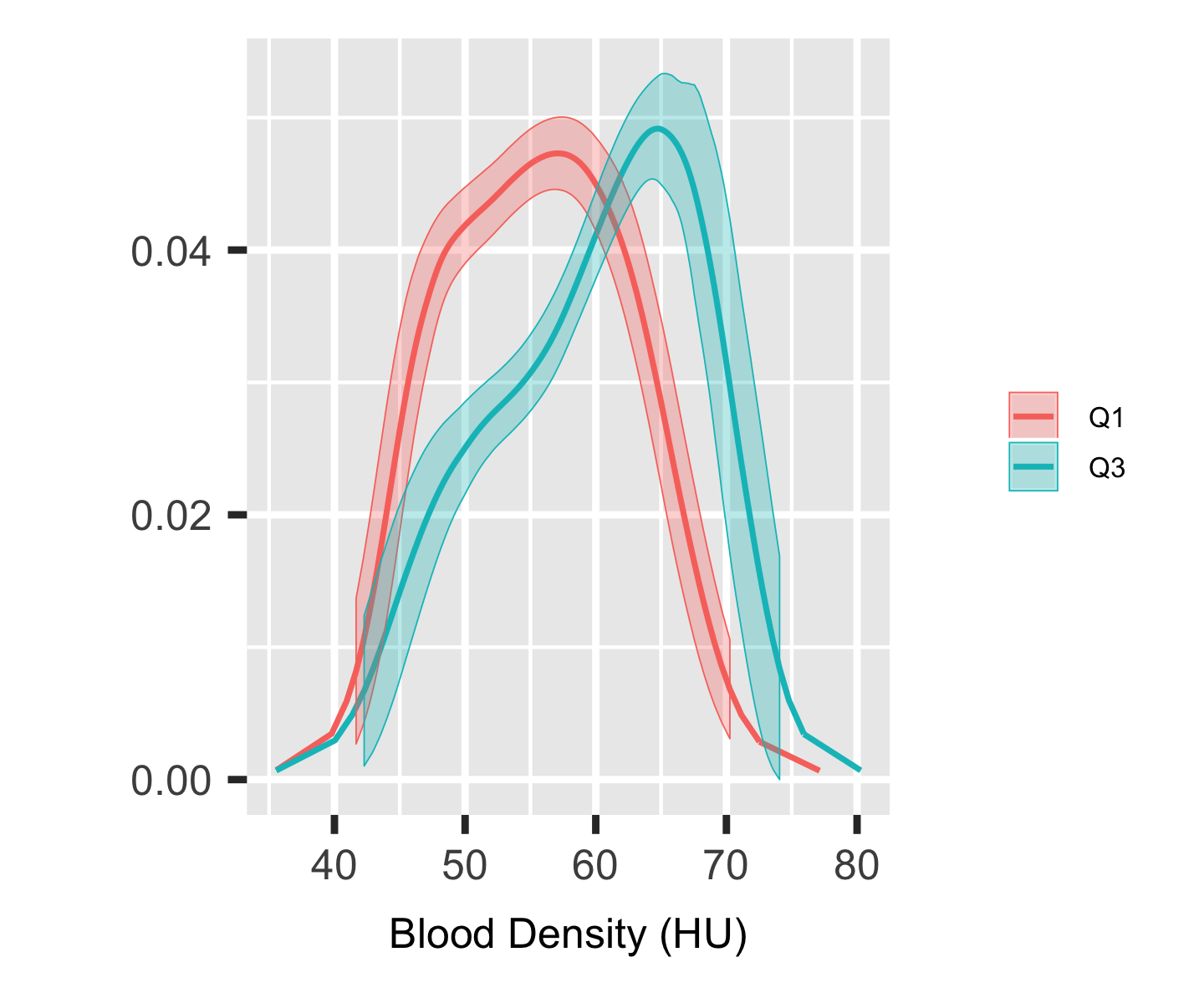}}
\caption{\footnotesize 95\% confidence bands for conditional Wasserstein mean CDF (left) and density (right) when hematoma volume is equal to the first (Q1) and third (Q3) quartile in the sample, with other variables set to their mean value for continuous variables and mode for binary variables. \label{fig: CB}}
\end{figure}

Finally, we demonstrate the use of the two types of Wasserstein confidence bands developed in Section~\ref{sec: cb}.  Figure~\ref{fig: CB_CDF} shows two fitted Wasserstein mean cdfs.  The first corresponds to a hematoma volume equivalent to the first quartile (Q1) of the observed values, with all other predictors set at their mean (for continuous variables) or mode (for binary variables).  The second is similar, but for the third quartile (Q3) of hematoma volume.  Each fitted cdf is also accompanied by a Wasserstein-$\infty$ band, represented by the stochastic ordering bracket $[F_L, F_U]$. These bands may be interpreted as bounds on the ``horizontal" sampling variability of the fitted distributions, i.e.\ the sampling variability of the fitted quantile function.  Figure~\ref{fig: CB_DENS} shows the corresponding fitted Wasserstein mean densities, along with density bands. These density bands reflect sampling variability at the density level and aid in inferring the relevance of local features seen in the fitted densities. As indicated by Figure~\ref{fig: ObsFitted}, the magnitude of the horizontal variability is the smaller of the two, and this is reflected in the confidence bands.  From a neurological point of view, the physiology dictates that larger hematomas tend to be more dense and homogeneous, since the pressure put on them from the surrounding tissue restricts growth and causes voids within the hematoma to be filled, confirming the observed differences between fitted cdfs/densities in Figure~\ref{fig: CB}.

\section{Discussion}

We have studied the regression of density response curves on vector predictors under the Fr\'echet regression model and the Wasserstein geometry of optimal transport.  The targets are the conditional Wasserstein mean densities, which can be estimated without the need of a tuning parameter, akin to a parametric model.  By replacing the additive error term in ordinary linear regression with a random optimal transport map, intuitive test statistics are proposed for testing null hypotheses of both no and partial predictor effects.  In the spirit of regression, asymptotic distributions are derived conditional on the observed predictors.  The covariance of the random transport map is a nuisance parameter that can be consistently estimated and thus used to form a rejection region with correct asymptotic size and which are uniformly consistent against classes of contiguous alternatives.

Confidence bands are also derived for the fitted Wasserstein mean densities in two forms.  Due to the intimate connection between the Wasserstein metric and quantile functions for univariate distributions, the first type of confidence band is formed in terms of the fitted Wasserstein mean quantile function (equivalently, the Wasserstein mean cdf), and forms a bracket in the usual stochastic ordering of distributions on the real line.  These intrinsic confidence bands are complemented by extrinsic bands in density space, allowing the user to simultaneously quantify sampling variability of all quantiles of the distribution, as well as the uncertainty of local features seen in the conditional Wasserstein mean densities. 

As for any regression model, it will be necessary in future work to develop diagnostic tools to assess the validity of the Wasserstein regression model for a given data set.  It is likely that tools for functional regression models can similarly be adapted to the setting of density response curves \cp{mull:07:3}.  Likewise, model selection or regularized estimation procedures are clearly desirable, especially in cases where the number of predictors is large \cp{barb:17}. 

While our theoretical developments have assumed the density response curves are known, in the majority of practical situations these will need to be estimated from a collection of univariate samples, each generated by one of the random densities.  In the simulations, we have demonstrated how this can be done using local linear smoothing of empirical quantile functions, while smoothed histograms could also be used as in the application to head CT densities.  Some previous theoretical work in the analysis of density samples has accounted for this preprocessing step, similar to the effects of pre-smoothing in functional data analysis when curves are measured sparsely in time and often contaminated with noise \cp{knei:01,mull:16:1,pana:16}.  An important point of future research in this area will include identifying a division of regimes between dense and sparse samples for density functions, similar to \ci{zhan:16} for classical functional data, and their implications on inferential procedures such as those proposed in this paper.

\section*{Acknowledgements}
The authors would like to thank Mostafa Jafari and Pascal Salazar for providing the hematoma density curves for our analysis.  


\newpage

\renewcommand{\theequation}{A.\arabic{equation}}
\renewcommand{\thesection}{A.\arabic{section}}

\begin{center}
{\Large APPENDIX}
\end{center}

\setcounter{section}{0}

The Appendix is organized as follows.  Section~\ref{sec: alg} gives algorithms for the various testing procedures and confidence band computations described in Sections~\ref{sec: tests} and \ref{sec: cb}.  Section~\ref{sec: fixed} illustrates how the confidence bands developed in Section~\ref{sec: cb} can be adapted to the case of random densities $\mk{F}$ with a fixed support $\I.$  Section~\ref{sec: sims2} gives simulation results, corresponding to the settings of Section~\ref{sec: sims}, for the case when the random densities are only observed through a random sample.  Section~\ref{sec: props} gives proofs of Propositions~\ref{prop: wMeanUnique} and \ref{prop: wanova}.  Sections~\ref{sec: testProofs} and \ref{sec: cbProofs} provide proofs of the testing and confidence band results in Sections~\ref{sec: tests} and \ref{sec: cb}, respectively.  Finally, Section~\ref{sec: lmas} gives statements and proofs of auxiliary lemmas.

\section{Computational Algorithms}
\label{sec: alg}

\subsection{Estimation}

This section gives two algorithms, one for computing $\hQp(x)$ in \eqref{eq: Qfit}, and the other for the density estimate $\hfp(x)$ in \eqref{eq: fullEst}.  Let $\ipLt{\cdot}{\cdot}$ and $\normLt{\cdot}$ denote the standard Hilbert inner product and norm on $L^2\zo.$  For any $Q \in \mk{Q},$ since $n\inv \son s_{in}(x) = 1,$ observe that
$$
\frac{1}{n}\son s_{in}(x)\normLt{Q_i - Q}^2 = \normLt{Q}^2 - 2\ipLt{\tilde{Q}_\oplus(x)}{Q_i} + \frac{1}{n}\son s_{in}(x) \normLt{Q_i}^2, 
$$
where $\tilde{Q}_\oplus(x) = \frac{1}{n}\son s_{in}(x)Q_i.$  For any $g,h \in L^2\zo,$ let $0=t_1 < \cdots < t_m = 1$ be a grid and set $\mathbf{g}_l = g(t_l),$ $\mathbf{h}_l = h(t_l).$  Let $A \in \R^{m\times m}$ be the matrix representing the trapezoidal rule approximation
$$
\i01 g(t) h(t) \d t \approx \mathbf{g}\T A \mathbf{h}.
$$

\begin{algorithm}[H]
\caption{Estimating $\hQp(x)$ \label{alg: Qest}}

\SetKwInOut{Input}{input}\SetKwInOut{Output}{output}

\Input{Predictor vector $x \in \R^p,$ quantile functions $Q_i$, and grid $0=t_1<\cdots<t_m=1$}
\Output{Estimates $\hQp(x,t_l),$ $l =1,\ldots,m$}
\BlankLine

\For{$l \leftarrow 1 $ \KwTo $m$}{
	$\mathbf{Q}_l  \leftarrow n\inv\son s_{in}(x) Q_i(t_l)$\;
}

\uIf{$\mathbf{Q}_{l+1} \geq \mathbf{Q}_l \,\, \forall l \in\{1,\ldots,m-1\}$}{
	$\hQp(x,t_l) \leftarrow \mathbf{Q}_l $\;
	}
\Else{
	$\mathbf{b}\ast \leftarrow \min_{\mathbf{b} \in \R^m} \frac{1}{2}\mathbf{b}\T A \mathbf{b} - \mathbf{Q}\T A \mathbf{b}, \quad \text{subject to } \mathbf{b}_1\leq \cdots \leq \mathbf{b}_m$\;
	
	$\hQp(x,t_l) \leftarrow \mathbf{b}_l\ast$\;
}
\end{algorithm}

\begin{algorithm}[H]
\caption{Estimating $\hfp(x)$ and $\hqp(x)$ \label{alg: fest}}

\SetKwInOut{Input}{input}\SetKwInOut{Output}{output}

\Input{Predictor vector $x \in \R^p,$ quantile functions values $Q_i(0)$, quantile densities $q_i$, grid $0=t_1<\cdots<t_m=1$, and $\epsilon > 0.$}
\Output{Estimates $\hqp(t_l)$ and $ \hfp(x,\hQp(x,t_l)),$ $l =1,\ldots,m$}
\BlankLine

$\tilde{\mathbf{Q}}_0 \leftarrow n\inv \son Q_i(0)$\;

\For{$l \leftarrow 1 $ \KwTo $m$}{
	$\mathbf{Q}_{l}  \leftarrow n\inv\son s_{in}(x) q_i(t_l)$\;
}

\If{$\mathbf{Q}_{l} < 0$ for any $ l \in\{1,\ldots,m\}$}{
	$\mathbf{b}\ast \leftarrow \min_{\mathbf{b} \in \R^{m+1}} \frac{1}{2}\mathbf{b}\T B \mathbf{b} - \mathbf{Q}\T B \mathbf{b}, \quad \text{such that } \mathbf{b}_2,\ldots, \mathbf{b}_{m+1} \geq \epsilon$\;
	
	$(\mathbf{Q}_1,\ldots,\mathbf{Q}_m) \leftarrow \mathbf{b}\ast$\;
}

$u_0 \leftarrow \mathbf{Q}_1$ \;

\For{$l \leftarrow 1$ \KwTo $m$}{
	$u_l \leftarrow u_{l-1} + (\mathbf{Q}_{l} + \mathbf{Q}_{l+1})(t_{l+1} - t_l)/2$\;

	$\hQp(x,t_l) \leftarrow u_l$\;
	
	$\hqp(x,t_l) \leftarrow \mathbf{Q}_l$\;
	
	$\hfp(x, u_l) \leftarrow 1  / \mathbf{Q}_{l} $\;

}
\end{algorithm}

Next, in light of Section~\ref{ss: WDB}, it is necessary to obtain a density estimate $\hfp(x).$  One possibility is to compute $\hQp(x)$ using Algorithm~\ref{alg: Qest}, followed by numerical inversion to obtain $\hFp(x,u),$ and finally computing $\hfp(x)$ by numerical differentiation.  However,  a more direct route is available.  Suppose that $g,h \in L^2\zo$ are differentiable, and take $\tilde{\mathbf{g}} = (g(0), g'(t_1),\ldots,g'(t_m)),$ and similarly $\tilde{\mathbf{h}}.$  Then let $B \in \R^{(m+1)\times(m+1)}$ be the matrix representing the trapezoidal approximation
$$
\i01 g(t) h(t) \d t \approx \tilde{\mathbf{g}}\T B \tilde{\mathbf{h}}.
$$

\subsection{Hypothesis Testing}

We first give the algorithm for computing the critical value $\hat{b}_\alpha^G$ in the global test using eigenvalue estimates $\hlambda_j.$  Assume the fitted quantile functions $\hat{Q}_i = \hQp(X_i)$ and $\hbQp = \hQp(\xbar)$ have already been computed using Algorithm~\ref{alg: Qest}, and that $\hat{C}_Q$ in \eqref{eq: CqEst} has also been calculated.  Then let $\hlambda_j$ be the eigenvalues of $\hat{C}_Q,$ which can be computed using a singular value decomposition of the discretized covariance.  Let $J$ be the number of $\lambda_j$ that are positive.  Algorithm~\ref{alg: crit} can be replaced by either of Algorithm~\ref{alg: critS} or \ref{alg: critB}, which correspond to the Satterthwaite approximation and residual bootstrap mentioned in Section~\ref{ss: testAlt}.  Similar algorithms can be followed for the partial tests of Section~\ref{ss: partial}.  However, the residual bootstrap algorithm can only be executed if the support of the random density $\mk{F}$ is fixed. 

\begin{algorithm}[H]
\caption{Critical Value for Global Test \label{alg: crit}}

\SetKwInOut{Input}{input}\SetKwInOut{Output}{output}

\Input{Eigenvalue estimates $\hlambda_j,$  level $\alpha \in (0,1)$, and integer $R$}
\Output{Critical value $\hat{b}_\alpha^G$}
\BlankLine

\For{$r \leftarrow 1 $ \KwTo $ R$}{
	
	Generate $\omega_j \overset{\textrm{i.i.d.}}{\sim} \chi^2_p,$ $j = 1,\ldots,J$\;
	
	$z_r \leftarrow \sum_{j = 1}^J \hlambda_j \omega_j$\;

}

$\hat{b}_\alpha^G \leftarrow 1-\alpha$ empirical quantile of $z_1,\ldots, z_R$\;

\end{algorithm}

\begin{algorithm}[H]
\caption{Satterthwaite Critical Value for Global Test \label{alg: critS}}

\SetKwInOut{Input}{input}\SetKwInOut{Output}{output}

\Input{Covariance Estimate $\hat{C}_Q$}
\Output{Critical value $\hat{b}_\alpha^G$}
\BlankLine

$a \leftarrow \left(\i01\i01 \hat{C}_Q(s,t) \d s \d t\right) \bigg / \left(\i01 \hat{C}_Q(t,t) \d t\right)$\;

$m \leftarrow p\left(\i01 \hat{C}_Q(t,t)\d t\right)^2 \bigg / \left(\i01 \i01 \hat{C}_Q(s,t) \d s \d t\right)$\;

$\hat{b}_\alpha^G \leftarrow a \chi^2_{m,1-\alpha}$\;

\end{algorithm}

\newpage 
\begin{algorithm}[H]
\caption{Bootstrap Critical Value for Global Test \label{alg: critB}}

\SetKwInOut{Input}{input}\SetKwInOut{Output}{output}

\Input{Data $(X_i, Q_i),$ fitted quantile functions $\hat{Q}_i$ and $\hbQp,$  level $\alpha \in (0,1)$, and integer $R$}
\Output{Critical value $\hat{b}_\alpha^G$}
\BlankLine

\For{$i \leftarrow 1$ \KwTo $n$}{

	$\hat{T}_i \leftarrow \hat{Q}_i \circ \hbFp$\;
	
}

\For{$r \leftarrow 1$ \KwTo $R$}{

	Sample $T_1\ast,\ldots T_n\ast$ uniformly and with replacement from $(\hat{T}_1,\ldots,\hat{T}_n)$\;
	
	$(Q_1\ast,\ldots,Q_n\ast) \leftarrow (T_1\ast \circ \bQp,\ldots, T_n\ast \circ \bQp)$\;
	
	Compute fitted values $\hat{Q}_{\oplus,r}(X_i)$ from data $(X_i, Q_i\ast)$ using Algorithm~\ref{alg: Qest}\;
	
	$z_r \leftarrow \son \normLt{\hat{Q}_{\oplus,r}(X_i) - \bQp}^2$\;

}

$\hat{b}_\alpha^G \leftarrow 1-\alpha$ empirical quantile of $z_1,\ldots, z_R$\;

\end{algorithm}

\subsection{Confidence Bands}

We first consider the Wasserstein-$\infty$ bands of Section~\ref{ss: IWB}.  The goal is to compute the upper and lower cdf bounds $(F_L, F_U)$ such that $C_{\alpha,n}(x) = \left[F_L, F_U\right].$  The initial step in the algorithm is to approximate the critical value $\hat{m}_\alpha$.  This requires generating zero-mean Gaussian processes (GP) with a specified covariance function.  This step can be done in a number of ways, for instance by discretizing the process onto a grid and generating a multivariate Gaussian vector with the given covariance structure.  However, other methods can be used.  The key step is to project the resulting functions $M_L$ and $M_U$ to their nearest monotonic majorant and minorant, respectively, which becomes a quadratic program when discretized.  For simplicity, we don't outline that step here, but the same approach used in Algorithm~\ref{alg: Qest} applies.  In addition, the algorithm produces quantile bounds $(Q_L, Q_U)$ instead of the cdf bounds since it is easier to present this way, but one can easily compute the latter through numerical inversion as a post-processing step.  Suppose that $\hat{D}_Q(s,t)$ has already been computed as in \eqref{eq: DqEst}, and define $\hat{\mathbf{D}}_{1,x}(s,t) = \tilde{x}\T \hat{D}_Q(s,t) \tilde{x}.$

\newpage

\begin{algorithm}[H]
\caption{Wasserstein-$\infty$ Band \label{alg: QBand}}

\SetKwInOut{Input}{input}\SetKwInOut{Output}{output}

\Input{Predictor vector $x \in \R^p,$ estimates $\hQp(x)$ and $\hat{\mathbf{D}}_{1,x}(s,t),$ level $\alpha \in (0,1)$, and integer $R$}
\Output{Lower/Upper bound quantile functions $(Q_{L}, Q_{U})$ of simultaneous Wasserstein-$\infty$ band $C_{\alpha,n}(x)$}
\BlankLine

\For{$r \leftarrow 1$ \KwTo $R$}{

	Generate GP $\hat{\mc{N}}_x(t)$ with mean zero and covariance $\hat{\mathbf{D}}_{1,x}(s,t)$\;
	
	$z_r \leftarrow \sup_{t \in \zo} \hat{\mathbf{D}}_{1,x}(t,t)^{-1/2} |\hat{\mc{N}}_x(t)|$\;

}

$\hat{m}_\alpha \leftarrow$ $1-\alpha$ empirical quantile of $z_1,\ldots,z_R$\;

$M_L(t) \leftarrow \hQp(x,t) - n^{-1/2}\hat{m}_\alpha \hat{\mathbf{D}}_{1,x}(t,t)^{1/2}$\;

$M_U(t) \leftarrow \hQp(x,t) + n^{-1/2}\hat{m}_\alpha \hat{\mathbf{D}}_{1,x}(t,t)^{1/2}$\;

$Q_L \leftarrow \argmin_{Q \in \mk{Q}} \normLt{Q - M_L}^2 \quad \text{subject to } Q \text{ nondecreasing, } Q \geq M_L$\;

$Q_U \leftarrow \argmin_{Q \in \mk{Q}} \normLt{Q - M_U}^2 \quad \text{subject to } Q \text{ nondecreasing, } Q \leq M_U$\;

\end{algorithm}

Lastly, we present the algorithm for computing the Wasserstein density band $B_{\alpha,n}(x)$ as in \eqref{eq: densCBd}.  Assume the covariance $\hat{\mathbb{D}}(s,t)$ in \eqref{eq: DmatEst} has been computed, using quantile estimates $\hat{Q}_i$, $\hQp(x)$ and quantile density estimates $\hat{q}_i$, $\hqp(x)$ from Algorithm~\ref{alg: fest}.  Furthermore, define 
$$
\hat{b}_t = \frac{1}{\hqp(x,t)^{3}}\left(\frac{\partial}{\partial t} \hqp(x, t),\, \hqp(x,t)\right)\T, \quad \frac{\partial}{\partial t} \hqp(x,t) = \frac{1}{n}\son s_{in}(x) q_i'(t).
$$ 
Then set $\hat{\mathbf{D}}_{2,x}(s,t) = \hat{b}_s\T\left(\tilde{x}\T \otimes \hat{\mathbb{D}}(s,t)\otimes \tilde{x}\right) \hat{b}_t.$  

\begin{algorithm}[H]
\caption{Wasserstein Density Band \label{alg: densBand}}

\SetKwInOut{Input}{input}\SetKwInOut{Output}{output}

\Input{Predictor vector $x \in \R^p,$ estimates $\hfp(x)$ and $\hat{\mathbf{D}}_{2,x}(s,t),$ level $\alpha \in (0,1)$, $\delta \in (0,1/2),$ and integer $R$}
\Output{Lower/Upper bound functions $(f_{\textrm{L}}, f_{\textrm{U}})$ of nearly simultaneous Wasserstein density band $B_{\alpha,n}(x)$}
\BlankLine

\For{$r \leftarrow 1 $ \KwTo $R$}{
	Generate GP $\hat{\mc{J}}(t)$ on $(\delta, 1-\delta)$ with mean zero and covariance $\hat{\mathbf{D}}_{2,x}(s,t)$\;
	
	$z_r \leftarrow \sup_{t \in (\delta,1-\delta)} \hat{\mathbf{D}}_{2,x}(t,t)^{-1/2} |\hat{\mc{J}}(t)|$
	}

$\hat{l}_\alpha \leftarrow 1-\alpha$ empirical quantile of $z_1,\ldots,z_R$\;

$f_{\textrm{L}}(u) = \hfp(x,u) - n^{-1/2}\hat{\mathbf{D}}_{2,x}(\hFp(x,u),\hFp(x,u))^{1/2}\hat{l}_\alpha$\;

$f_{\textrm{U}}(u) = \hfp(x,u) + n^{-1/2}\hat{\mathbf{D}}_{2,x}(\hFp(x,u),\hFp(x,u))^{1/2}\hat{l}_\alpha$\;
\end{algorithm}

\section{Densities with Fixed Support}
\label{sec: fixed}

In many common scenarios, the sample of densities all have a known finite support $\I = (a,b).$  Without loss of generality, suppose $\I = (0,1).$  The key difficulty associated with the Wasserstein-$\infty$ band in \eqref{eq: WinfCB} is the standardization provided obtained by dividing by the square root of the pointwise variance estimate.  When the densities have fixed support, the random optimal transports $T$ necessarily satisfy $C_T(0,0) = C_T(1,1) = 0,$ so that standardization does not preserve tightness in the weak convergence statements.  On the other hand, having a fixed support resolves the issue encountered in the development of the Wasserstein density bands, since the true conditional Wasserstein mean $\fp(x)$ and its estimate $\hfp(x)$ will both have support $(0,1).$  

Using the same notation as in Section~\ref{sec: cb}, for any $\delta \in (0,1/2),$ define 
\[
\begin{split}
\hat{\zeta}_x' &= \sup_{t \in [\delta, 1-\delta]} \left[\tilde{x}\T \hat{D}_Q(t,t) \tilde{x}\right]^{-1/2} \left|\hat{\mc{N}}_x(t)\right|, \\
\hat{\xi}_x' &= \sup_{t \in (0,1)} \left[ \hat{\mc{R}}_x(\hQp(x,t), \hQp(x,t)) \right]^{-1/2} \left| \hat{\mc{J}}_x(t) \right|,
\end{split}
\]
and let $\hat{m}_\alpha'$ and $\hat{l}_\alpha'$ be, respectively, the $1-\alpha$ quantiles of $\hat{\zeta}_z'$ and $\hat{\xi}_x'.$  Then define the confidence bands
\begin{align}
C'_{\alpha,n}(x) &= \left\{ g \in \mc{D}: \, \sup_{u \in [\delta, 1-\delta]} \frac{|T_{g,x}(u) - u |}{\hat{\mc{C}}_x(u,u)^{1/2}} \leq  \frac{\hat{m}_\alpha'}{n^{1/2}}\right\}, \quad T_{g,x} = G\inv \circ \hFp(x), \\
B'_{\alpha,n}(x) &= \left\{g \in \mc{D}: \, \sup_{u \in \I} \frac{|g(u) - \hfp(x,u)|}{\hat{\mc{R}}_x(u,u)^{1/2}} \leq \frac{\hat{l}_\alpha'}{n^{1/2}}\right\}.
\end{align}

The same computational procedures outlined in Algorithms~\ref{alg: QBand} and \ref{alg: densBand} apply to computing the above confidence bands.  We also have strongly consistent coverage.

\begin{Corollary}
\label{cor: CB_fixed}
Suppose that (T1)--(T4) and $C_T(u,u) > 0$ for all $u \in (0,1).$  Then
$$
\PX\left(\fp(x) \in C'_{\alpha,n}(x)\right) \rightarrow 1-\alpha \quad \text{almost surely}.
$$
\end{Corollary}

\begin{Corollary}
\label{cor: WDB_fixed}
Suppose the assumptions of Corollary~\ref{cor: WDB} hold.  Then, if $\inf_{u \in (0,1)} \mc{R}_x(u,u) > 0,$
$$
\PX\left(\fp(x) \in B'_{\alpha,n}(x)\right) \rightarrow 1-\alpha \quad \text{almost surely}.
$$
\end{Corollary}

\section{Simulations with Indirectly Observed Densities}
\label{sec: sims2}

To understand the possible errors that may be introduced when densities are only observed through a random sample, an additional step was added to the density simulations outlined in Section~\ref{sec: sims}.  After generating the random densities $\mk{F}_i,$ $i = 1,\ldots,n,$ a random sample of scalar variables was generated for each density, each of size 300.  The preliminary estimation of $Q_i,$ $q_i = Q_i'$ and $q_i'$ was performed as described in Section~\ref{sec: sims}.

\begin{figure}[t]
\centering
\subcaptionbox{}[5in]{\includegraphics[scale = 0.9]{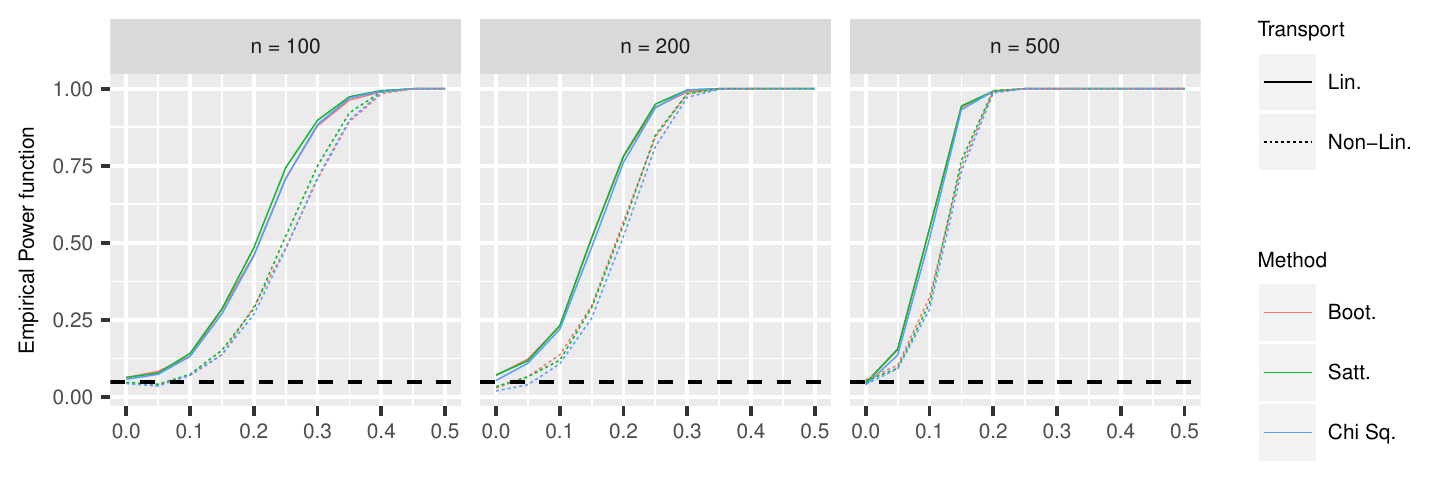}} \\
\subcaptionbox{}[5in]{\includegraphics[scale = 0.9]{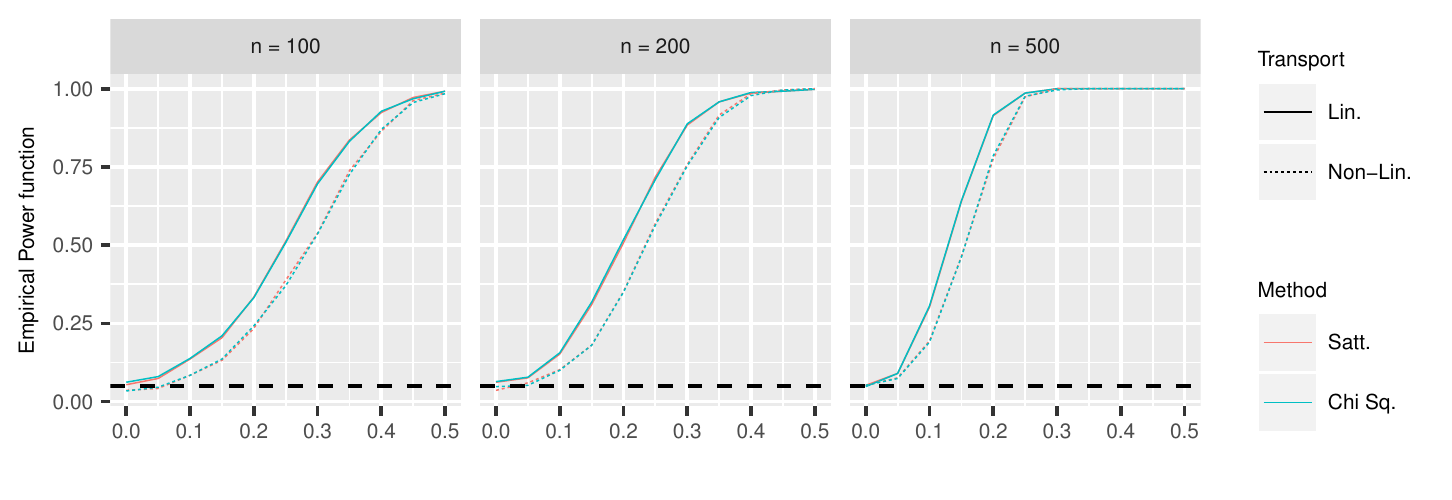}}
\caption{\small Power curves for global (top) and partial (bottom) Wasserstein $F$-tests, for sample sizes $n = 100, 200, 500,$ when densities are only indirectly observed through a random sample. The dotted horizontal line represents the nominal level $\alpha = 0.05.$ \label{fig: test_est}}
\end{figure}

\begin{table}[H]
    \centering
    \caption{Error rates of intrinsic Wasserstein-$\infty$ and Wasserstein density bands for indirectly observed densities over 500 simulation runs}
    \begin{tabular}{l|ccccccc}
\toprule
& \multicolumn{6}{c}{Wasserstein-$\infty$ band*} \\
        &  \multicolumn{3}{c}{linear transport map} & & \multicolumn{3}{c}{nonlinear transport map}  \\\cline{2-4}\cline{6-8}
& n = 100 & n = 200 & n = 500 &  & n = 100  & n = 200 & n = 500 \\ 
\hline 
x = -0.30 & 0.072 & 0.086 & 0.072 & &  0.066 & 0.076 & 0.078 \\ 
x = -0.24 & 0.072 & 0.080 & 0.070  & & 0.078 & 0.070 & 0.084 \\ 
x = -0.18 & 0.072 & 0.078 & 0.088  & & 0.092 & 0.064 & 0.082 \\ 
x = -0.12 & 0.076 & 0.074 & 0.094 &  & 0.084 & 0.062 & 0.084 \\ 
x = -0.06 & 0.072 & 0.070 & 0.088 &  & 0.088 & 0.056 & 0.088 \\ 
x = 0 & 0.066 & 0.082 & 0.092 &  & 0.082 & 0.058 & 0.086 \\ 
x = 0.06 & 0.062 & 0.076 & 0.076 & &  0.084 & 0.062 & 0.082 \\ 
x = 0.12 & 0.062 & 0.068 & 0.082 &  & 0.098 & 0.072 & 0.088 \\ 
x = 0.18 & 0.062 & 0.064 & 0.068 &  & 0.090 & 0.078 & 0.092 \\ 
x = 0.24 & 0.070 & 0.060 & 0.066 &  & 0.094 & 0.088 & 0.090 \\ 
x = 0.30 & 0.064 & 0.064 & 0.072 & &  0.096 & 0.072 & 0.096 \\
\hline
& \multicolumn{6}{c}{Wasserstein density band*} \\
        &  \multicolumn{3}{c}{linear transport map} & & \multicolumn{3}{c}{nonlinear transport map}  \\\cline{2-4}\cline{6-8}
& n = 100 & n = 200 & n = 500 &  & n = 100  & n = 200 & n = 500 \\ 
\hline
x = -0.3 & 0.272 & 0.188 & 0.114 & &  0.072 & 0.048 & 0.044 \\ 
x = -0.24 & 0.228 & 0.148 & 0.100 &  & 0.068 & 0.046 & 0.052 \\ 
x = -0.18 & 0.206 & 0.146 & 0.096  & & 0.050 & 0.048 & 0.054 \\ 
x = -0.12 & 0.182 & 0.120 & 0.098 & &  0.032 & 0.042 & 0.042 \\ 
x = -0.06 & 0.166 & 0.112 & 0.088 &  & 0.030 & 0.034 & 0.038 \\ 
x = 0 & 0.150 & 0.106 & 0.084 & &  0.024 & 0.026 & 0.032 \\ 
x = 0.06 & 0.140 & 0.098 & 0.074  & & 0.030 & 0.028 & 0.034 \\ 
x = 0.12 & 0.134 & 0.112 & 0.068  & & 0.030 & 0.030 & 0.030 \\ 
x = 0.18 & 0.152 & 0.118 & 0.064 &  & 0.036 & 0.032 & 0.034 \\ 
x = 0.24 & 0.160 & 0.116 & 0.060  & & 0.048 & 0.030 & 0.042 \\ 
x = 0.3 & 0.176 & 0.116 & 0.062 &  & 0.044 & 0.040 & 0.040\\
 \bottomrule  
\end{tabular}
\begin{tablenotes}
\centering
      \item *$\delta = 0.1$ used to avoid boundary issue.
    \end{tablenotes}
    \label{tab: cb_est}
\end{table}

\section{Proofs of Propositions~\ref{prop: wMeanUnique} and \ref{prop: wanova}}
\label{sec: props}

\begin{proof}[Proof of Proposition~\ref{prop: wMeanUnique}]
By (A1), $\bQp(t) := E(Q(t))$ is well-defined and finite for any $t \in (0,1).$  Since $Q$ is continuous almost surely by (A2), for any $t \in (0,1)$ and sequence $t_k \rightarrow t,$ $Q(t_k) \rightarrow Q(t)$ almost surely.  Then, by dominated convergence, $\bQp(t_k) \rightarrow \bQp(t).$  By monotonic convergence, we can extend $\bQp$ to $\zo,$ with $\bQp(t) \uparrow \bQp(1)$ as $t \uparrow 1$ and $\bQp(t) \downarrow \bQp(0)$ as $t \downarrow 0.$   Clearly, $\bQp$ is increasing by montonicity of expectation, so that $\bQp$ is a valid quantile function.  Furthermore, $\bQp \in L^2\zo$ by (A1). In fact, for any measure $\mk{M}$ on $\R$ with finite second moment, its quantile function $\mc{Q}$ must be in $L^2[0,1].$  Denote the standard Hilbert norm on this space by $\normLt{\cdot}.$  Thus, the measure corresponding to $\bQp$ is the unique minimizer of $E\left[\normLt{Q - \mc{Q}}^2\right]$ among all quantile functions $\mc{Q} \in L^2\zo,$ so that $\bQp$ represents the quantile function of the Wasserstein mean measure.

Let $t \in (0,1)$ and $\delta$ as in (A2), and recall that $q(t) = Q'(t) = 1/\left[f\circ Q(t)\right] \in (0,\infty)$ by (A2).  If $t_k \rightarrow t$ with $|t_k - t|< \delta,$ (A2) implies that $$\left|\frac{Q(t) - Q(t_k)}{t_k - t}\right| \leq \sup_{|s-t| < \delta} q(s).$$  So, by dominated convergence, we have $\bqp(t) := E(q(t)) = {\bQp}'(t).$  As $0 < \bqp(t) < \infty$ for all $t \in (0,1),$ $\bFp(u) = {\bQp}\inv(u)$ is a bona fide cdf.  Furthermore, $\bFp$ is differentiable with density $\bfp(u) = 1/\left[\bqp\circ \bFp(u)\right]$ for $u \in (\bQp(0), \bQp(1)).$  That $\bfp$ is the unique minimizer of \eqref{eq: wMeanVar} follows from the fact that, for any $g\in \mc{D},$ $E\left[\dw^2(\mk{F},g)\right] = E\left[\normLt{Q - G\inv}^2\right],$ so that $\Vp(\mk{F}) < \infty.$  Lastly, by (A3), there exists $M > 0$ such that $P\left( \sup_{u \in (Q(0), Q(1))} \mk{F}(u) < M\right) > 0.$  Thus, for $t \in (0,1),$
$$
\bqp(t) \geq E\left(\inf_{t \in (0,1)} q(t)\right)  \geq \frac{P\left(\sup_{u \in (Q(0),Q(1))} \mk{F}(u) < M\right)}{M} > 0,
$$
so that $\sup_{u \in (\bQp(0), \bQp(1)} \bfp(u) < \infty.$

The results for conditional means are obtained similarly and the details are omitted.
\end{proof}

\begin{proof}[Proof of Proposition~\ref{prop: wanova}]
Under the Wasserstein regression model,
$$
\Qp(x) = \Pi_{\mc{Q}}\left(E\left[s(X,x)Q(\cdot)\right]\right)
$$
where, for any continuous function $h$ on $\zo,$ $\Pi_{\mc{Q}}(h)$ is the projection of $h$ onto
$$
\mk{Q}:= \left\{\mc{Q} \in L^2[0,1]:\, \mc{Q} \textrm{ is a quantile function}\right\}.
$$
As $\mk{Q}$ is closed and convex, this projection exists and is unique \cp{rych:12}; in fact, it is continuous \cp{groe:10,lin:14}. Under the assumptions, $\Qp(x,t)$ is continuous and strictly increasing for almost all $x,$ so that the projection operator is redundant and, by the form of $s(X,x)$,
\begin{equation}
\label{eq: WmodelQ}
\Qp(x,t) = \bQp(t) + E\left[(X - \mu)\T Q(t)\right]\Sigma\inv(x - \mu).
\end{equation}
This implies that $E(\Qp(X,t)) = \bQp(t)$ for all $t \in [0,1],$ so that $E(S(u)) = E(\Qp(X) \circ \bFp(u)) = u$ for almost all $u \in \Istc = [\bQp(0), \bQp(1)].$ 

For the variance decomposition, since $Q = T \circ \Fp(X) = T\circ S \circ \bFp,$
\begin{equation}
\label{eq: decomp}
\begin{split}
\Vp(\mk{F}) &= E\left[\int_\Ist \left(T \circ S(u) - u\right)^2\bfp(u)\,\d u\right] \\
&=\int_\Ist E\left[(T\circ S(u) - S(u))^2 + (S(u) - u)^2\right. \\&\hspace{1cm}  \left.+ 2(T\circ S(u) - S(u))(S(u) - u)\right]\bfp(u)\,\d u \\
&=\int_\Ist E\left[(T\circ S(u) - S(u))^2\right]\bfp(u)\,\d u + \int_\I E\left[(S(u) - u)^2\right] \bfp(u)\, \d u,
\end{split}
\end{equation}
and the cross product term vanishes since $E(T\circ S(u) | X) = S(u)$ a.s.  Because $E(S(u)) = u$, $\bfp$ is the Wasserstein mean of the random density $\fp(X),$ and the second term in the last line of \eqref{eq: decomp} is $E(\dw^2(\fp(X), \bfp) = \Vp(\fp(X))$ as claimed.  Furthermore, by (T1), $E\left[(T\circ S(u) - S(u))^2|X \right] = C_T(S(u), S(u))$ almost surely for all $u\in \Istc.$  The results follows.
\end{proof}

\section{Proofs of Theorems~\ref{thm: global}--\ref{thm: Ppower} and Corollaries~\ref{cor: Gsize} and \ref{cor: Psize}}
\label{sec: testProofs}

We first introduce some notation.  For $x \in \R^p,$ set $\tilde{x} = (1\,\, x\T)\T \in \R^{p+1},$ and let $\Lambda = E(\tilde{X}\tilde{X}\T),$ $\hat{\Lambda} = n\inv \son \tilde{X}_i \tilde{X}_i\T.$ Define $\Ist = (\bQp(0), \bQp(1))$ and $\Istc$ its closure.  For $u \in \Istc,$ let $R = Q\circ \bFp$, $R_i = Q_i \circ \bFp,$  and define
\begin{equation}
\label{eq: pi_gamma}
\begin{split}
\gamma(u) &= \Cov(X, R(u)), \quad \tilde{\gamma}(u) = \frac{1}{n}\son X_i R_i(u), \\
\pi(u) &= \Lambda\inv E\left[\tilde{X}R(u)\right], \quad \tilde{\pi}(u) = \hLambda\inv\left[\frac{1}{n} \son \tilde{X}_iR_i(u)\right]
\end{split}
\end{equation}    
For $u' \in \Ist,$ let $\pi'(u')$ and $\tilde{\pi}'(u')$ denote the derivatives of $\pi$ and $\tilde{\pi},$ respectively.  The symbol $\otimes$ will denote the Kronecker product of matrices.

For any set $\mc{T}$ and $k \in \mathbb{N},$ define $L_k^\infty(\mc{T})$ as the $k$-fold Cartesian product of all bounded functions on $\mc{T}.$ A statistic $W$ computed from the data $(X_i, \mk{F})$ is said to be $\opx(1)$ if, for all $\epsilon > 0,$ $\PX(|W| > \epsilon)$ converges to zero almost surely.  Likewise, for a nonnegative sequence $a_n,$ $W = \Opx(a_n)$ if
$$
\limsup_{C \rightarrow \infty} \lim_{n \rightarrow \infty} \PX\left(|W| > Ca_n\right) = 0 \quad \text{almost surely}.
$$

\begin{proof}[Proof of Theorem~\ref{thm: global}]
Define $\tilde{\pi}$ as in \eqref{eq: pi_gamma}, and set
\begin{equation}
\label{eq: tQp}
\begin{split}
\tilde{Q}_\oplus(x,t) &= \frac{1}{n}\son s_{in}(x) Q_i(t) 
= \tilde{x}\T \tilde{\pi} \circ \bQp(t), \quad t \in \zo, \\
\tilde{q}_\oplus(x,t) &= \frac{\partial}{\partial t} \tilde{Q}_\oplus(x,t) = \frac{1}{n}\son s_{in}(x)q_i(t) 
= \bqp(t)\tilde{x}\T \tilde{\pi}' \circ \bQp(t) , \quad t \in (0,1),
\end{split}
\end{equation}
as empirical versions of $\Qp(x,t) = \tilde{x}\T \pi \circ \bQp(t),$ $\qp(x,t) = \bqp(t) \tilde{x}\T \pi \circ \bQp(t).$ Let $\Pi_{\mc{Q}}$ be as defined in the proof of Proposition~\ref{prop: wanova}.  Then, for each $i,$ $$\dw(\hfi, \hfp\ast) = \left\lVert \Pi_{\mc{Q}}\left(\tilde{Q}_\oplus(X_i,t)\right) - \hQp\ast\right\rVert_{L^2}, \quad \hbQp(t) = \frac{1}{n}\son Q_i(t).$$
By Lemma~\ref{lma: regular}, the event 
\begin{equation}
\label{eq: An}
A_n = \left\{\min_{1 \leq i \leq n} \inf_{t \in (0,1)} \tilde{q}_\oplus(X_i,t)> 0\right\}
\end{equation}
satisfies $\PX(A_n) \rightarrow 1$ almost surely.  Inded, when $A_n$ holds, $\tilde{Q}_\oplus(X_i,t)$ is strictly increasing for all $i$ and $\Pi_{\mc{Q}}$ is redundant.  

Let $\gamma$ and $\tilde{\gamma}$ be as in \eqref{eq: pi_gamma}.  Since $\hbQp(t) = \tilde{Q}_\oplus(\xbar,t),$ whenever $A_n$ holds, the test statistic becomes
\begin{equation*}
\begin{split}
F\ast_G &= \son \dw^2(\hfi, \hfp\ast) = \son \i01 \left[\tilde{Q}_\oplus(X_i,t) - \tilde{Q}_\oplus(\xbar,t)\right]^2\d t \\
&= \son \int_\Ist \left[(X_i - \xbar)\T\hSigma\inv \tilde{\gamma}(u)\right]^2 \, \bfp(u) \d u \\
&=n\int_{\Ist} \tilde{\gamma}\T(u)\hSigma\inv \tilde{\gamma}(u)\, \bfp(u)\d u,
\end{split}
\end{equation*}
using the change of variables $t = \bFp(u).$

From \eqref{eq: WmodelQ}, we deduce that $S_i(u) = \tilde{X}_i\T \pi(u).$  Hence, when $\mc{H}_G$ holds, $S = \id$ almost surely, so that $\gamma(u) \equiv 0.$ Then, almost surely, by Lemma~\ref{lma: Qjoint}, $\sqrt{n}\hSigma\inv\tilde{\gamma} |X_1,\ldots,X_n$ converges weakly to a zero-mean Gaussian process $\Xi$ on $L_p^\infty(\Istc)$ with covariance $\Cov(\Xi(u), \Xi(v)) = \Sigma\inv C_T(u,v)$.  Applying the continuous mapping theorem,
$$
F\ast_G | X_1,\ldots,X_n \overset{D}{\rightarrow} \int_{\Istc} \Xi\T(u) \Sigma \,\Xi(u) \bfp(u) \d u
$$
almost surely.  The result follows since $K(u,v) = C_T(u,v)$ under $\mc{H}_0^G,$ so that with $(\lambda_k, \phi_k)$ as in \eqref{eq: Keigen} and $\tilde{\Xi} = \Sigma^{1/2}\Xi,$
\[
\int_{\Istc} \normE{\tilde{\Xi}(u)}^2\bfp(u)\d u  
= \sum_{k = 1}^\infty \lambda_k \left[\sum_{j = 1}^p Z_{jk}^2\right],
\]
where $Z_{jk} = \lambda_k^{-1/2}\int_{\Istc}\tilde{\Xi}_j(u)\phi_k(u)\bfp(u)\d u$ are standard Gaussian random variables, independent across both $j$ and $k.$
\end{proof}


\begin{proof}[Proof of Corollary~\ref{cor: Gsize}]
Define $\tilde{Q}_\oplus(x,t)$ as in \eqref{eq: tQp}, and set
$$
\tilde{C}_Q(s,t) = \frac{1}{n}\son (Q_i(s) - \tilde{Q}_\oplus(X_i,s))(Q_i(t) - \tilde{Q}_\oplus(X_i,t)).
$$
Let $Z,$ $Z_n$ be (conditional on the data) zero-mean Gaussian processes on $L_p\zo$ with covariance functions $I_p \tilde{C}_Q(s,t)$ and $I_p C_Q(s,t),$ respectively, where $I_p$ is the $p\times p$ identity matrix.  That $\tilde{C}_Q(s,t)$ converges almost surely to $C_Q(s,t)$ for any $s,t \in \zo$ follows by standard arguments, using assumptions (T1) and (T3).  Furthermore, using (T2)--(T4), it follows that $|\tilde{C}_Q(s,t) - \tilde{C}_Q(s,t')| = O(1)|t-t'|$ almost surely, where the $O(1)$ term is uniform over $s,t,t' \in \zo.$  Thus, Lemma~\ref{lma: GP} implies that $Z_n \rsa Z$ in $L_p^\infty\zo$ almost surely.  

Let $\omega_j$ be as in the statement of Theorem~\ref{thm: global}.  If $\tilde{\lambda}_j$ are the eigenvalues of $\tilde{C}_Q,$ observe that $\sum_{j = 1}^\infty \lambda_j \omega_j \overset{D}= \i01 \normE{Z(t)}^2 \d t$ and $\sum_{j = 1}^\infty \tilde{\lambda}_j\omega_j \overset{D}= \i01 \normE{Z_n(t)}^2\d t$ almost surely. Let $\tilde{b}_\alpha^G$ be the $1-\alpha$ quantile of $\sum_{j = 1}^\infty \tilde{\lambda}_j \omega_j.$  Then, by the continuous mapping theorem, $\tilde{b}_\alpha^G \rightarrow b_\alpha^G$ almost surely.

Let $A_n$ be as in \eqref{eq: An}, so that $\hat{C}_Q = \tilde{C}_Q$ when $A_n$ holds.  Then, by Lemma~\ref{lma: regular}, for any $\epsilon > 0,$
\[
\begin{split}
\beta_n^G &= \PX\left(F_G\ast > \hat{b}_\alpha^G\right) \\
&\leq \PX\left(F_G\ast > b_\alpha^G - \epsilon\right) + \PX(A_n^c) + \PX\left(|\tilde{b}_\alpha^G - b_\alpha^G| > \epsilon\right) \\
&= P\left(\sum_{j = 1}^\infty \lambda_j \omega_j > b_\alpha^G - \epsilon\right) + o(1) \quad \text{almost surely}.
\end{split}
\]
Since $\epsilon$ was arbitrary, $\limsup_{n \rightarrow \infty} \beta_n^G \leq \alpha$ almost surely.  A similar lower bound shows $\liminf_{n\rightarrow \infty} \beta_n^G \geq \alpha$ almost surely, proving the result.
\end{proof}

\begin{proof}[Proof of Theorem~\ref{thm: Gpower}]
Let $\epsilon > 0$ be given.  Let $A_n$ be as in \eqref{eq: An}, and set 
\begin{equation}
\label{eq: rhodef}
\rho(u) = \Sigma\inv \gamma(u), \quad \hat{\rho}(u) = \hSigma\inv \tilde{\gamma}(u).  
\end{equation}
For any model $\mc{G} \in \mk{G}_{A,n},$ define 
\begin{equation}
\label{eq: Vdef}
\begin{split}
W &= W_{\mc{G}} = E\left[\dw^2(\fp(X), \bfp)\right] = \int_\Ist \rho\T(u)\Sigma\rho(u)\bfp(u) \d u, \\
\hat{W}_n &= \int_\Ist \rho\T(u) \hSigma \rho(u) \bfp(u) \d u,
\end{split}
\end{equation}
and set
\begin{equation}
\label{eq: Delta}
\Delta_n = \left|2\int_{\Ist} \left[\hat{\rho}(u) - \rho(u)\right]\T \hSigma \rho(u)\bfp(u)\d u +\hat{W}_n - W\right|.
\end{equation}
Lastly, define
\[
\label{eq: rj}
\begin{split}
r_{1n} &= \sup_{\mc{G} \in \mk{G}_{A,n}}P\left(\sup_{m \geq n} \PX\left( \Delta_m > \frac{W}{2}\right) \geq \epsilon\right), \\
r_{2n} &= \sup_{\mc{G} \in \mk{G}_{A,n}} P\left(\sup_{m \geq n} \PX(A_m^c) \geq \epsilon\right) \\
r_{3n} &= \sup_{\mc{G} \in \mk{G}_{A,n}} P\left(\sup_{m \geq n} \PX\left( \hat{b}_\alpha^G > R\right) \geq \epsilon\right) 
\end{split}
\]
for some $R>0.$  It is shown in Lemma~\ref{lma: rj} that, for any $\epsilon > 0,$ $r_{1n}$ and $r_{2n}$ are both $o(1)$ and that there exists $R > 0$ such that $r_{3n} = o(1).$

Next, on the set $A_n,$ one can write the test statistic $F_G\ast $ as
\[
\begin{split}
F_G\ast &= n\int_\Ist (\hat{\rho}(u) - \rho(u))\T\hSigma(\hat{\rho}(u) - \rho(u)) \bfp(u) \d u \\
&\hspace{0.5cm} + 2n\int_\Ist (\hat{\rho}(u) - \rho(u))\T\hSigma \rho(u) \bfp(u) \d u + n \hat{W}_n.
\end{split}
\]
Then, for $n_0$ large enough that $nW \geq na_n^2 > 2R$ for all $n \geq n_0$ and $\mc{G} \in \mk{G}_{A,n}$, for such $n$ we immediately obtain
\[
\begin{split}
\beta_n^G &= \PX(F_G\ast > \hat{b}_\alpha) \geq \PX(F_G\ast > R) - \PX(\hat{b}_\alpha^G > R) \\
&\geq \PX\left(n \int_\Ist (\hat{\rho}(u) - \rho(u))\T\hSigma(\hat{\rho}(u) - \rho(u)) \bfp(u) \d u > R - n\left[W -\Delta_n\right]\right) \\
&\hspace{2cm} - \PX(A_n^c) - \PX(\hat{b}_\alpha^G > R) \\
&\geq 1 - \PX\left(\Delta_n > \frac{W}{2}\right) - \PX(A_n^c) - \PX(\hat{b}_\alpha^G > R).
\end{split}
\]
As a result, for any $\epsilon > 0,$ there exists $R > 0$ such that, for large $n,$
\[
\inf_{\mc{G} \in \mk{G}_{A,n}} P\left(\inf_{m\geq n} \beta_m^G \geq 1 - 3\epsilon\right) \geq 1 - r_{1n} - r_{2n} - r_{3n} = 1 - o(1).
\]

\end{proof}


\begin{proof}[Proof of Theorem~\ref{thm: partial}]

In analogy to \eqref{eq: tQp}, define
\begin{equation}
\label{eq: tQpz}
\begin{split}
\tilde{Q}_{0,\oplus}(y,t) &= \frac{1}{n}\son s_{in,0}(y) Q_i(t) 
= \tilde{y}\T \tpiy \circ \bQp(t), \quad t \in \zo, \\
\tilde{q}_{0,\oplus}(y,t) &= \frac{\partial}{\partial t} \tilde{Q}_{0,\oplus}(y,t) = \frac{1}{n}\son s_{in,0}(y)q_i(t) 
= \bqp(t)\tilde{y}\T \tpiy' \circ \bQp(t) , \quad t \in (0,1),
\end{split}
\end{equation}
as empirical versions of 
\[
\begin{split}
Q_{0,\oplus}(y,t) &= \tilde{y}\T \piy \circ \bQp(t), \\
q_{0,\oplus}(y,t) &= \bqp(t) \tilde{y}\T \piy \circ \bQp(t).
\end{split}
\]
Note that $\hQp(t) = \ybar\T\tpiy \circ \bQp(t) = \xbar\T \tilde{\pi} \circ \bQp(t).$

By a similar argument as in the proof of Theorem~\ref{thm: global}, letting
\begin{equation}
\label{eq: An0}
A_{n,0} = \left\{ \min_{1 \leq i \leq n} \inf_{t \in (0,1)} \tilde{q}_{0,\oplus}(Y_i, t) > 0\right\},
\end{equation}
then $\PX(A_{n,0}) \rightarrow 1$ almost surely under $\mc{H}_0^G.$  With $\tilde{\gamma}$ as in \eqref{eq: pi_gamma}, partition it as $(\tilde{\gamma}_Y, \tilde{\gamma}_Z).$  When both $A_n$ and $A_{n,0}$ hold, 
\[
\begin{split}
\dw^2(\hfi, \hfp) &= \int_\Ist \left[(X_i - \xbar)\T\hSigma\inv \tilde{\gamma}(u)\right]^2 \bfp(u)\d u, \\
\dw^2(\hfzi, \hfp) &= \int_\Ist \left[(Y_i - \ybar)\T\hSigyy\inv \tgammay(u)\right]^2 \bfp(u)\d u.
\end{split}
\]
Letting $$\hat{\Gamma} = \left(\begin{array}{c c} I_q & 0 \\ \hSigzy \hSigyy\inv & 0\end{array}\right),$$ it can be verified that 
$$
\hSigma\inv \hat{\Gamma} = \begin{pmatrix}
\hSigyy\inv & 0 \\ 0 & 0
\end{pmatrix}, \quad
(I - \hat{\Gamma})\hSigma = \begin{pmatrix}
0 & 0 \\ 0 & \hSigma_{Z|Y}
\end{pmatrix}.
$$

Next, define $\rho,\, \hat{\rho}$ as in \eqref{eq: rhodef}, and partition them as $\rho = (\rho_Y , \rho_Z),$ $\hat{\rho} = (\hat{\rho}_Y, \hat{\rho}_Z).$  Then, under $\mc{H}_0^P,$ $\rho_Z(u) \equiv 0.$  By straightforward algebra, one can verify that, under $\mc{H}_0^P$ and on the set $A_n \cap A_{n,0},$ the test statistic becomes
\begin{equation*}
\begin{split}
F\ast_P &=  \son \int_\Ist \left\{ \left[(X_i - \xbar)\T \hSigma\inv \tilde{\gamma}(u)\right]^2 - \left[(Y_i - \ybar)\T \hSigyy\inv \tilde{\gamma}_Y(u)\right]^2\right\} \bfp(u) \d u \\
&= n \int_\Ist \left[\tilde{\gamma}\T(u) \hSigma\inv \tilde{\gamma}(u) - \tilde{\gamma}\T(u)\hSigma\inv \hat{\Gamma}\tilde{\gamma}(u)\right] \bfp(u) \d u \\
&=n \int_\Ist \hat{\rho}\T(u) (I - \hat{\Gamma})\hSigma \hat{\rho}(u) \bfp(u) \d u \\
&= n\int_\Ist \hat{\rho}_Z(u)\T \hSigma_{Z|Y}\hat{\rho}_Z(u) \bfp(u) \d u.
\end{split}
\end{equation*}

Thus, under $\mc{H}_0^P,$ by applying Lemma~\ref{lma: Qjoint}, the continuous mapping theorem, and the law of large numbers, we see that
$$
F\ast_P | X_1,\ldots,X_n \overset{D}{\rightarrow} \int_{\I} \normE{\tilde{\Xi}(u)}^2 \bfp(u) \d u
$$
almost surely, where $\tilde{\Xi}(u)$ is a zero-mean Gaussian process on $L^\infty_{p-q}(\Ist)$ with covariance
\begin{equation}
\label{eq: partialCov}
\Sigma_{Z|Y}^{-1/2}J\T E\left[(X - \mu)(X - \mu)\T C_T(S(u), S(v))\right]J\Sigma_{Z|Y}^{-1/2},
\end{equation}
where $J\T = (-\Sigzy\Sigyy\inv \, I_{p-q}).$  By our assumptions, this is the kernel of a self-adjoint, trace-class operator, and so has nonnegative eigenvalues $\tau_j$ with $\sum_{j = 1}^\infty \tau_j < \infty.$  By the Karhunen-Lo\`eve representation,
$$
 \int_{\I} \normE{\tilde{\Xi}(u)}^2 \bfp(u) \d u \overset{D}{=} \sum_{j = 1}^\infty \tau_j \xi_j^2,
$$
where $\xi_j$ are i.i.d.\ standard normal random variables.

Under the additional assumptions that $E(Z|Y)$ is linear in $Y$ and $\Var(Z|Y)$ is constant, it follows that $J\T(X - \mu) = Z - E(Z|Y)$ and $\Var(J\T(X - \mu)|Y) = \Sigma_{Z|Y}$. Consequently, the covariance of $\tilde{\Xi}(u)$ becomes
\[
\Sigma_{Z|Y}^{-1/2}J\T E\left[(X - \mu)(X - \mu)\T C_T(S(u), S(v))\right]J\Sigma_{Z|Y}^{-1/2} = I_{p-q}K(u,v)
\]
and the result follows.
\end{proof}


\begin{proof}[Proof of Corollary~\ref{cor: Psize}]
The proof is similar to that of Corollary~\ref{cor: Gsize} and is omitted.
\end{proof}

\begin{proof}[Proof of Theorem~\ref{thm: Ppower}]
The proof is similar to that of Theorem~\ref{thm: Gpower} and is omitted.
\end{proof}

\section{Proofs of Theorems~\ref{thm: fittedVal} and \ref{thm: densCLT} and Corollaries~\ref{cor: CB}--\ref{cor: WDB_fixed}}
\label{sec: cbProofs}


\begin{proof}[Proof of Theorem~\ref{thm: fittedVal}]
First, note that $\Qp(x)\circ\bFp = \tilde{x}\T\pi$.  With $\tilde{q}_\oplus$ defined in \eqref{eq: tQp}, whenever $\inf_{t \in (0,1)} \tilde{q}_\oplus(x,t) > 0,$ we have $\hQp(x) \circ \bFp = \tilde{x}\T \tilde{\pi},$ implying that
\[
\begin{split}
\sqrt{n}\left(\hat{V}_x(u) - u\right) &= \sqrt{n}\left(\hQp(x) - \Qp(x)\right)\circ \Fp(x) \\
&= \sqrt{n}\tilde{x}\T\left(\tilde{\pi} - \pi\right)\circ \bQp \circ \Fp(x).
\end{split}
\]
Let $\mc{S}_1$ be the first $p+1$ coordinates of the multivariate Gaussian process $\mc{S}$ defined in Lemma~\ref{lma: Qjoint}, so that
$$
\Cov(\mc{S}_1(u), \mc{S}_1(v)) = \Lambda\inv E\left(\tilde{X}\tilde{X}\T C_T(S(u), S(v))\right) \Lambda\inv = \tilde{K}(u,v), \quad u,v \in \Istc.
$$ 
Then, by the continuous mapping theorem and Lemma~\ref{lma: regular},
$$
\sqrt{n}(\hat{V}_x - \mathrm{id})|X_1,\ldots,X_n \rsa \tilde{x}\T\mc{S}_1\circ \bQp \circ \Fp(x)=: \mc{M}_x \quad \text{almost surely,}
$$
where $\mc{M}_x$ is a zero mean Gaussian process on $L^\infty(\I_x).$  For $u,v \in \I_x,$ set $u_x,$ $v_x$ as in the statement of the theorem.  Then
\begin{equation}
\label{eq: CxCov}
\begin{split}
\Cov(\mc{M}_x(u), \mc{M}_x(v)) &= \tilde{x}\T\tilde{K}(\bQp \circ \bFp(x,u), \bQp\circ\bFp(x,v)) \tilde{x} \\
&= \tilde{x}\T \tilde{K}(u_x, v_x)\tilde{x}.
\end{split}
\end{equation}
\end{proof}

\begin{proof}[Proof of Corollary~\ref{cor: CB}]
Define $\tilde{Q}_\oplus(x,t)$ as in \eqref{eq: tQp}, and set
\[
\begin{split}
B_Q(s,t) &= E\left[\tilde{X}\tilde{X}\T C_T(\Qp(x,s), \Qp(x,t))\right] \\
\tilde{B}_Q(s,t) &=\frac{1}{n}\son \tilde{X}_i\tilde{X}_i\T (Q_i(s) - \tilde{Q}_\oplus(X_i,s))(Q_i(t) - \tilde{Q}_\oplus(X_i,t)), \\
\tilde{D}_Q(s,t) &=\hLambda\inv \tilde{B}_Q(s,t)\hLambda\inv.
\end{split}
\]
In a similar way to the proof of Corollary~\ref{cor: Gsize}, when can show that $$\normF{\tilde{B}_Q(s,t) - B_Q(s,t)}$$ converges to zero almost surely, for any $s,t \in \zo,$ and that $$\normF{\tilde{B}_Q(s,t) - \tilde{B}_Q(s,t')} = O\left(|t-t'|\right)$$ almost surely, where the $O(1)$ term is uniform over $s,t,t' \in \zo.$  Let $\tilde{\mc{N}}_x,$ $\mc{N}_x$ be zero-mean Gaussian processes on $L^\infty\zo$ with covariance $\tilde{x}\T\tilde{D}_Q(s,t)\tilde{x}$ and $\tilde{x}\T D_Q(s,t) \tilde{x},$ respectively.  Then, by Lemma~\ref{lma: GP} and Slutsky's lemma, $\tilde{\mc{N}}_x \rsa \mc{N}_x$ almost surely in $L^\infty\zo.$  Hence, if $\tilde{m}_\alpha$ is the $1-\alpha$ quantile of $$\sup_{t \in \zo}\left[\tilde{x}\T\tilde{D}_Q(t,t)\right]^{-1/2} |\tilde{\mc{N}}_x(t)|,$$ we have $\tilde{m}_\alpha \rightarrow m_\alpha$  almost surely.  However, with $A_n$ as in \eqref{eq: An}, $\hat{m}_\alpha = \tilde{m}_\alpha$ whenever $A_n$ holds, so that $\hat{m}_\alpha - m_{\alpha} = \opx(1)$ by Lemma~\ref{lma: regular}.

Note that the above argument also shows that $\sup_{t \in \zo} \normF{\tilde{D}_Q(t,t) - D_Q(t,t)}$ converges to zero almost surely.  By the assumption that $\inf_{u \in \Ist} C_T(u,u) > 0,$ it is clear that $\inf_{t \in \zo} \tilde{x}\T D_Q(t,t)\tilde{x} > 0.$  Hence, since $\tilde{D}_Q = \hat{D}_Q$ on $A_n,$ we have
$$
\sup_{u \in \I_x} \left| \frac{\hat{\mc{C}}_x(\hat{V}_x(u), \hat{V}_x(u))}{\mc{C}_x(u,u)} -1 \right| = \sup_{t \in \zo} \left|\frac{\tilde{x}\T \hat{D}_Q(t,t)\tilde{x}}{\tilde{x}\T D_Q(t,t) \tilde{x}} - \right| = \opx(1).
$$
Then, for any $\epsilon,\eta > 0,$
\[
\begin{split}
\liminf_n & \PX\left(\fp(x) \in C_{\alpha,n}(x)\right)  \\& = \liminf_n \PX\left(\sup_{v \in \hat{\I}_x} \frac{\sqrt{n}\left|\hat{V}\inv_x(v) -v\right|}{\hat{\mc{C}}_x(v,v)^{1/2}} \leq \hat{m}_\alpha \right) \\
&\geq \liminf_{n} \PX\left(\sup_{u \in \I_x} \frac{\sqrt{n}|\hat{V}_x(u) - u|}{\mc{C}_x(u,u)^{1/2}} \leq \hat{m}_\alpha(1 - \epsilon)\right) \\
&\hspace{1cm} - \limsup_n \PX\left(\inf_{t \in \zo} \frac{\hat{\mc{C}}_x(\hat{V}_x(u), \hat{V}_x(u))}{\mc{C}_x(u,u)} > (1-\epsilon)^2\right)\\
&\geq \liminf_{n} \PX\left(\sup_{u \in \I_x} \frac{\sqrt{n}|\hat{V}_x(u) - u|}{\mc{C}_x(u,u)^{1/2}} \leq (m_\alpha - \eta)(1-\epsilon)\right)  \\
&\hspace{1cm} - \limsup_n \PX\left(|\hat{m}_\alpha - m_\alpha| > \eta\right) +o(1) \\
& = P\left(\zeta_x \leq (m_\alpha - \eta)(1 - \epsilon)\right) + o(1),
\end{split}
\]
almost surely.  Letting $\epsilon, \eta \rightarrow 0$ yields $\liminf_n\PX\left(\fp(x) \in C_{\alpha,n}(x)\right) \geq 1-\alpha$ almost surely.  A similar upper bound yields the result.
\end{proof}


\begin{proof}[Proof of Theorem~\ref{thm: densCLT}]
Let $\mc{S} = (\mc{S}_1,\mc{S}_2)$ be as in Lemma~\ref{lma: Qjoint}.  Now, define the map $\psi_x: L^\infty_{p+1} (\Istc) \times L^\infty_{p+1}(\Ist) \rightarrow L^\infty\zo \times L^\infty(0,1)$ by $$\psi_x(h_1,h_2) = (\tilde{x}\T h_1\circ \bQp, \tilde{x}\T (h_2 \circ \bQp)\bqp).$$  With $\tilde{Q}_\oplus(x)$ and $\tilde{q}_\oplus(x)$ as defined in \eqref{eq: tQp}, and $\pi$ and $\tilde{\pi}$ as in \eqref{eq: pi_gamma}, we have $\psi_x(\tilde{\pi}, \tilde{\pi}') = (\tilde{Q}_\oplus(x), \tilde{q}_\oplus(x))$ and $\psi_x(\pi,\pi') = (\Qp(x), \qp(x)).$  Then, by the continuous mapping theorem,
$$
\sqrt{n}((\tilde{Q}_\oplus(x), \tilde{q}_\oplus(x)) - (\Qp(x), \qp(x))) | X_1,\ldots,X_n \rsa \psi_x(\mc{S}_1,\mc{S}_2) =: \mc{T}_x
$$
in $L^\infty\zo \times L^\infty(0,1)$ almost surely.  Here, $\mc{T}_x(t,t') = (\mc{T}_{x1}(t), \mc{T}_{x2}(t'))$ is a zero-mean Gaussian process with covariance
\begin{equation}
\label{eq: TxCov}
\Cov(\mc{T}_x(s,s'), \mc{T}_x(t,t')) = E\left[\left(\tilde{x}\T\Lambda\inv \tilde{X}\right) \tilde{\mc{L}}(s,s', t,t') \left(\tilde{x}\T\Lambda\inv \tilde{X}\right)\right],
\end{equation}
where, letting $U_s = \Qp(X,s)$,
$$
\tilde{\mc{L}}(s,s',t,t') = \begin{pmatrix}
C_T(U_s, U_t) & C_T^{(0,1)}(U_s, U_{t'})\qp(X, t') \\
\qp(X, s')C_T^{(1,0)}(U_{s'}, U_t) & \qp(X, s')C_T^{(1,1)}(U_{s'}, U_{t'})\qp(X, t')
\end{pmatrix}.
$$
In particular, if $\mathbb{K}$ and $\mathbb{D}$ are defined as in Section~\ref{ss: WDB}, then for $s,t \in (0,1),$ 
$$
\Cov(\mc{T}_x(s,s), \mc{T}_x(t,t)) = \tilde{x}\otimes \mathbb{K}(\bQp(s), \bQp(t)) \otimes \tilde{x} = \tilde{x}\T \otimes \mathbb{D}(s,t), \otimes \tilde{x}.
$$

Observe that $\fp(x) = 1/\left[\qp(x) \circ \Fp(x)\right]$ and that, whenever $\inf_{t \in (0,1)} \tilde{q}_\oplus(x,t) > 0,$ $\hfp(x) = 1/\left[\tilde{q}_\oplus(x) \circ \tilde{F}_\oplus(x)\right]$ is well-defined, where $\tilde{F}_\oplus(x) = \tilde{Q}_\oplus(x) \inv.$  Moreover, by the assumption that $\bfp$ is continuously differentiable, the same holds for $\fp(x).$  Now, by Lemmas~\ref{lma: regular} and \ref{lma: hadamard}, the delta method applies, so that
$$
\sqrt{n}(\hfp(x) - \fp(x))|X_1,\ldots,X_n \rsa \left[\frac{\partial}{\partial u} \fp(x)\right]\mc{T}_{x1}\circ\Fp(x) - \left[\fp(x)^2\right]\mc{T}_{x2}\circ \Fp(x) =: \mc{F}_x
$$
in $L^\infty(\I_x^\delta)$ almost surely.  With $c_u = ((\partial / \partial u)\fp(x, u), -\fp(x,u)^2)\T$ and $u_x = \bQp \circ \Fp(x,u),$ the covariance of $\mc{F}_x$ is
\[
\begin{split}
\Cov(\mc{F}_x(u), \mc{F}_x(v)) &= c_u\T \Cov(\mc{T}_x(\Fp(x,u), \Fp(x,u)), \mc{T}_x(\Fp(x,v), \Fp(x,v))) c_v \\
&= c_u\T \left[ \tilde{x}\T \otimes \mathbb{K}(u_x, v_x) \otimes \tilde{x}\right] c_v.
\end{split}
\]
\end{proof}

\begin{proof}[Proof of Corollary~\ref{cor: WDB}]
The proof is similar to that of Corollary~\ref{cor: CB}, using the additional assumptions in Theorem~\ref{thm: fittedVal}, and is omitted.
\end{proof}

\begin{proof}[Proof of Corollary~\ref{cor: CB_fixed}]
By (T1), for any $\delta > 0,$ $\inf_{u \in [\delta, 1-\delta]} C_T(u,u) > 0.$  The remainder of the proof follows in the same was as that of Corollary~\ref{cor: CB}.
\end{proof}

\begin{proof}[Proof of Corollary~\ref{cor: WDB_fixed}]
Under the assumption of fixed support, Theorem~\ref{thm: densCLT} can be adapted to show that
$$
\sqrt{n}(\hfp(x) - \fp(x)) | X_1,\ldots,X_n \rsa \mc{V}_x'
$$
in $L^\infty(0,1)$ almost surely.  Then, since $\inf_{u \in (0,1)} \mc{R}_x(u,u) > 0,$ one can follow the steps of the proof of Corollary~\ref{cor: CB} to show the coverage result.
\end{proof}

\section{Auxiliary Lemmas}
\label{sec: lmas}


\begin{Lemma}
\label{lma: Qjoint}
Assume (T1)--(T4) hold.  Then there exists a zero-mean Gaussian process $\mc{S}$ on $L_{p+1}^\infty(\Istc) \times L_{p+1}^\infty(\Ist)$ such that
$$
\sqrt{n}\left(\begin{pmatrix} \tilde{\pi} \\ \tilde{\pi}'\end{pmatrix} - \begin{pmatrix} \pi \\ \pi'\end{pmatrix} \right)|X_1,\ldots,X_n \rsa \mc{S}
$$
almost surely.  Furthermore, for $u,v \in \Istc$ and $u',v' \in \Ist$,
\begin{equation*}
\Cov(\mc{S}(u,u'), \mc{S}(v,v')) =  E\left[\left(\Lambda\inv\tilde{X}\right) \otimes \mc{L}(u,u', v,v') \otimes \left(\Lambda\inv\tilde{X}\right)\T\right],
\end{equation*}
where
$$
\mc{L}(u,u',v,v') =  \begin{pmatrix} C_T(S(u), S(v)) & C_T^{(0,1)}(S(u), S(v'))S'(v') \\ S'(u')C_T^{(1,0)}(S(u'), S(v)) & S'(u')C_T^{(1,1)}(S(u'),S(v'))S'(v') \end{pmatrix}.
$$

\end{Lemma}

\begin{proof}
The proof follows from a multivariate extension (see Problem 1.5.3 of \ci{well:96}) of Example 2.11.13 in \ci{well:96}, which is itself a generalization of a result by \ci{jain:75}.  Letting $\hat{\Lambda} = n\inv\son\tilde{X}_i\tilde{X}_i\T,$ for $(u,u') \in \Istc \times \Ist,$ set $Z_{ni}(u,u') = n^{-1/2}\left(\hLambda\inv \tilde{X}_i\right) \otimes (R_i(u), R_i'(u'))\T.$  Then
$$
\sqrt{n}\begin{pmatrix}
\tilde{\pi}(u) \\ \tilde{\pi}'(u')
\end{pmatrix}
= \son Z_{ni}(u,u'), \quad 
\sqrt{n}\begin{pmatrix}
{\pi}(u) \\ {\pi}'(u')
\end{pmatrix}
= \son \EX\left(Z_{ni}(u,u')\right).
$$
Using assumption (T2) and (T4), one can derive the bounds
\begin{equation}
\label{eq: Rbounds}
\begin{split}
|R_i(u) - R_i(v)| &\leq C_1 \left(\sup_{u \in \R} |T_i'(u)|\right) \normE{\tilde{X}_i} |u-v|, \\
|R_i'(u') - R_i'(v')| &\leq C_2\left[L_i \normE{\tilde{X}_i} + \sup_{u \in \R} |T_i'(u)|\right]\normE{\tilde{X}_i}|u'-v'|,
\end{split}
\end{equation}
where $L_i$ is the Lipschitz constant of $T_i'$ and $C_1,C_2$ are constants.  Hence, setting 
$$
M_{ni} = n^{-1/2}\normF{\hLambda \inv}\normE{\tilde{X}_i}^2\left[\normE{\tilde{X}_i}^2 L_i^2 + \sup_{u \in \R} |T'(u)|^2\right]^{1/2},
$$
we obtain $ \normE{Z_{ni}(u,u') - Z_{ni}(v,v')} \leq C M_{ni}\normE{(u,u') - (v,v')}$ for some $C > 0.$  By (T3), 
\begin{align*}
\son \EX(M_{ni}^2)  &= \frac{1}{n} \son \tilde{X}_i\T \hLambda^{-2} \tilde{X_i} \normE{\tilde{X}_i}^2\left[\normE{\tilde{X}_i}^2 L_i^2 + \sup_{u \in \R}|T'(u)|^2\right] \\
&\leq O(1)\cdot \frac{1}{n} \son \left[\normE{\tilde{X}_i}^6 L_i^2 + \normE{\tilde{X}_i}^4 \sup_{u \in \R} |T'(u)|^2\right] \\
&= O(1)
\end{align*}
almost surely.  

Next, we verify the uniform Lindeberg condition.  By monotonicity of $T$ and assumptions (T1) and (T4), 
\begin{align*}
\label{eq: Rsup}
\sup_{u \in \Istc} |R_i(u)|^2 &\leq \sup_{u \in \R} C_T(u,u) + C_3\normE{\tilde{X}_i}^2, \\
\sup_{u' \in \Ist} |R_i'(u')|^2 &\leq C_4 \normE{\tilde{X}_i}^2 \sup_{u' \in \Ist} |T_i'(u)|^2 
\end{align*}
Then, by (T4), 
\begin{equation}
\label{eq: Rfin}
E\left(\normE{\tilde{X}_i}^2\left[\sup_{u \in \Istc} |R(u)|^2 + \sup_{u' \in \Ist} |R'(u')|^2\right]\right) < \infty.
\end{equation}
By a Borel-Cantelli argument, $\max_{1 \leq i \leq n} \normE{\tilde{X}_i} = o(n^{1/2})$ almost surely.  Hence, for any $\epsilon, M > 0,$ almost surely for large $n$ we will have $$\max_{1 \leq i \leq n} \normE{\hLambda\inv \tilde{X}_i}^2 < \frac{n\epsilon^2}{M}.$$  Hence, for some constant $C$,
\begin{align*}
\limsup_{n \rightarrow \infty} &\son \EX\left[\sup_{u \in \Istc,\, u' \in \Ist} \lVert Z_{ni}(u,u')\rVert_E^2 \mathbf{1}\left(\sup_{u \in \Istc,\, u' \in \Ist} \normE{Z_{ni}(u,u')} > \epsilon\right)\right] \\
&\leq \limsup_{n \rightarrow \infty} \frac{O(1)}{n}\son \normE{\tilde{X}_i}^2 \EX\left[ \left(\sup_{u \in \Istc} |R_i(u)|^2 + \sup_{u' \in \Ist} |R_i'(u)|^2\right) \right. \\
& \hspace{2.5cm} \times  \left .\mathbf{1}\left(\sup_{u \in \Istc} |R_i(u)|^2 + \sup_{u' \in \Ist} |R_i'(u)|^2 > M\right)\right] \\
&\leq C\limsup_{n \rightarrow \infty} \frac{1}{n}\son \normE{\tilde{X}_i}^2 \EX\left[ \left(\sup_{u \in \Istc} |R_i(u)|^2 + \sup_{u' \in \Ist} |R_i'(u)|^2\right) \right. \\ 
&\hspace{2.5cm} \times \left. \mathbf{1}\left(\sup_{u \in \Istc} |R_i(u)|^2 + \sup_{u' \in \Ist} |R_i'(u)|^2 > M\right)\right] \\
&= CE\left[\normE{\tilde{X}}^2  \left(\sup_{u \in \Istc} |R(u)|^2 + \sup_{u' \in \Ist} |R'(u)|^2\right) \mathbf{1}\left(\sup_{u \in \Istc} |R(u)|^2 + \sup_{u' \in \Ist} |R'(u)|^2 > M\right)\right]
\end{align*}
almost surely.  By \eqref{eq: Rfin} and dominated convergence, the last line converges to 0 as $M \rightarrow \infty.$

Lastly, we demonstrate pointwise convergence of the covariance.  By (T1), we have $\Cov(R_i(u), R_i(v)|X_i) = C_T(S_i(u), S_i(v) ).$  Furthermore, using (T2), (T3), and the conditional dominated convergence theorem, 
\begin{align*}
\Cov(R_i(u), R_i'(v') | X_i) &= C_T^{(0,1)}(S_i(u), S_i(v))S_i'(v') , \\ 
\Cov(R_i'(u'), R_i'(v)|'X_i) &= S_i'(u')C_T^{(1,1)}(S_i(u), S_i(v))S_i'(v').
\end{align*}
Hence, by the law of large numbers,
\begin{align*}
&\CX\left(\son Z_{ni}(u,u'), \sum_{k = 1}^n Z_{nk}(v,v')\right) \\
& \hspace{1cm} =  \left[ \frac{1}{n} \son \left(\hLambda\inv\tilde{X}_i\right) \otimes \begin{pmatrix} \Cov(R_i(u), R_i(v)|X_i) & \Cov(R_i(u), R_i'(v')|X_i) \\ \Cov(R_i'(u'), R_i(v)|X_i) & \Cov(R_i'(u'), R_i'(v')|X_i) \end{pmatrix} \otimes \left(\hLambda\inv\tilde{X}_i\right)\T \right]  \\
& \hspace{1cm} \rightarrow  E\left[  \left(\Lambda\inv\tilde{X}\right) \otimes \begin{pmatrix} C_T(S(u), S(v)) & C_T^{(0,1)}(S(u), S(v'))S'(v') \\ S'(u')C_T^{(1,0)}(S(u'), S(v)) & S'(u')C_T^{(1,1)}(S'(u'), S'(v'))S'(v') \end{pmatrix} \otimes \left(\Lambda\inv\tilde{X}\right)\T \right] 
\end{align*}	
almost surely.
\end{proof}


\begin{Lemma}
\label{lma: regular}
For any $x \in \R^p,$ define $\tilde{q}_\oplus(x,t) = n\inv \son s_{in}(x)q_i(t).$  Then, under (T1)--(T4), 
$$
\PX\left(\inf_{t \in (0,1)} \tilde{q}_\oplus(x,t) > 0\right), \quad \PX\left(\min_{1 \leq i \leq n} \inf_{t \in (0,1)} \tilde{q}_\oplus(X_i,t) > 0 \right)
$$
both converge to 1 almost surely.
\end{Lemma}

\begin{proof}
Observe that $\tilde{q}_\oplus(x,t) = \tilde{x}\T \tilde{\pi}'(\bQp(t)) \bqp(t)$ and $\qp(x,t) = \tilde{x}\T \pi'(\bQp(t))\bqp(t).$  With $M$ as in (T4), $\inf_{t \in (0,1)} \tilde{q}_\oplus(x,t) > M\inv - \sup_{t \in (0,1)} |\tilde{q}_\oplus(x,t) - \qp(x,t)|.$  But, by (T4) and Lemma~\ref{lma: Qjoint},
$$
\sup_{t \in (0,1)} |\tilde{q}_\oplus(x,t) - \qp(x,t)| \leq \normE{\tilde{x}}\left[ \sup_{t \in (0,1)} \bqp(t)\right] \sup_{u \in \Ist}\normE{\tilde{\pi}'(u) - \pi'(u)} = \Opx(n^{-1/2}) = \opx(1),
$$
proving the first claim.  Since $E(\normE{X}^2) < \infty,$ a Borel-Cantelli argument shows that $\max_{1 \leq i \leq n} \normE{\tilde{X}_i} = o(n^{1/2}) = \opx(n^{1/2})$ almost surely.  Hence,
\[
\begin{split}
\min_{1 \leq i \leq n} \inf_{t \in (0,1)} \tilde{q}_\oplus(X_i,t) &> M\inv - \max_{1 \leq i \leq n} \normE{\tilde{X}_i}\left[ \sup_{t \in (0,1)} \bqp(t)\right] \sup_{u \in \Ist} \normE{\tilde{\pi}'(u) - \pi'(u)} \\
&= M\inv - \opx(n^{1/2})\Opx(n^{-1/2}) = M\inv - \opx(1),
\end{split}
\]
proving the second claim.
\end{proof}


\begin{Lemma}
\label{lma: GP}
Let $\mc{T} \subset \R$ be a bounded set, and suppose $\mc{M},$ $\mc{M}_n,$ are zero-mean (possibly multivariate) Gaussian process on $L^\infty(\mc{T}),$ with non-degenerate covariances functions $C$ and $C_n,$ respectively.  Suppose that
\begin{enumerate}[i)]
\item $C$ is uniformly continuous,
\item $C_n \rightarrow C$ pointwise, and
\item $|C_n(s,t) - C_n(s,t')| = O(|t-t'|),$ uniformly in $s,t,t'.$
\end{enumerate}
Then $\mc{M}_n \rightarrow \mc{M}$ in $L^\infty(\mc{T}).$
\end{Lemma}

\begin{proof}
We show the result for the univariate case.  The extension to multivariate processes is straightforward.  Clearly, $\mc{M}_n$ converges marginally to $\mc{M},$ so it remains to show tightness.

For any $\eta > 0,$ choose $\epsilon > 0$ such that $|C(s,t) - C(s,t')| < \eta$ whenever $|t-t'| < \epsilon.$  Let $t_1,\ldots,t_L \in \mc{T}$ be such that, for any $t \in \mc{T},$ $|t - t_l|<\epsilon$ for some $l = 1,\ldots,L = O(\epsilon\inv).$ For any $l = 1,\ldots,L$ and $t \in \mc{T}$, let $t_m$ satisfy $|t - t_m| < \epsilon.$  Then
$$
|C_n(t_l,t) - C(t_l,t)| \leq O(\epsilon) + \max_{m = 1,\ldots,L} |C_n(t_l,t_m) - C(t_l,t_m)| + \eta.
$$  
It follows that $\max_{l = 1,\ldots,L} \sup_{t \in \mc{T}} |C_n(t_l, t) - C(t_l, t)| = o(1).$  A similar bound yields $\sup_{s,t} |C_n(s,t) - C_n(s,t)| = o(1).$

For $s,t \in \mc{T},$ let 
\[
\begin{split}
\rho(s,t) &= \left[C(s,s) + C(t,t) - 2C(s,t)\right]^{1/2} \\
\rho_n(s,t) &= \left[C_n(s,s) + C_n(t,t) - 2C_n(s,t)\right]^{1/2}
\end{split}
\]
be the standard deviation metrics of $\mc{M}$ and $\mc{M}_n,$ respectively, so that $\sup_{s,t} |\rho_n(s,t) - \rho(s,t)| = o(1).$  For $\delta > 0,$ let $n$ be large enough that $\sup_{s,t \in \mc{T}} |\rho_n(s,t) - \rho(s,t)| < \delta/2$ and $ \sup_{|s-t| < \epsilon} \rho_n(s,t) < D\epsilon^{1/2}$ for some $D > 0.$  Then, using Corollary 2.2.8 of \ci{well:96}, 
\[
\begin{split}
E\left(\sup_{\rho(s,t) < \delta/2} |\mc{M}_n(s) - \mc{M}_n(t)|\right) & \leq E\left(\sup_{\rho_n(s,t) < \delta} |\mc{M}_n(s) - \mc{M}_n(t)| \right)\\
&= O\left(\int_0^\delta \sqrt{-\log(\tau)} \d \tau\right) = o(1)
\end{split}
\]
as $\delta \rightarrow 0.$  Thus, $\mc{M}_n$ is asymptotically $\rho$-equicontinuous in probability, and thus asymptotically tight.
\end{proof}


\begin{Lemma}
\label{lma: rj}
For any $\epsilon > 0$ and some $R > 0,$ let $r_{jn}$, $j = 1,2,3,$ be as defined as in \eqref{eq: rj}.  Then, under the assumption of Theorem~\ref{thm: Gpower}, there exists $R>0$ such that $r_{jn} = o(1)$, $ j = 1,2,3.$
\end{Lemma}
\begin{proof}
Begin with $j =1,$ and define $W,$ $\hat{W}_n$ as in \eqref{eq: Vdef}.  Then
$$
\left| \int_\Ist(\hat{\rho}(u) - \rho(u))\T\hSigma \rho(u) \bfp(u) \d u\right| \leq \hat{W}_n^{1/2} \left[\int_\Ist (\hat{\rho}(u) - \rho(u))\T\hSigma (\hat{\rho}(u) - \rho(u)) \bfp(u) \d u\right]^{1/2}.
$$
Next, a direct calculation gives $\EX(\hat{\rho}(u)) = \rho(u)$ and $$\VX(\hat{\rho}(u)) = \hSigma\inv\left(\frac{1}{n^2} \son (X_i - \xbar)(X_i - \xbar)\T C_T(S_i(u),S_i(v))\right)\hSigma\inv,$$
so that, with $C_1 = \sup_{\mc{G} \in \mk{G}\ast} \sup_{u \in \Ist} C_T(u,u) < \infty,$
\[
\begin{split}
&\EX\left[\left\{\int_\Ist(\hat{\rho}(u) - \rho(u))\T\hSigma\rho(u) \bfp(u) \d u\right\}^2\right] \leq\hat{W}_n\mathrm{tr}\left(\hSigma \int_\Ist \VX(\hat{\rho}(u)) \bfp(u) \d u\right) \\
&\hspace{1cm} = \frac{\hat{W}_n}{n^2}\son \left\{\mathrm{tr}\left[(X_i - \xbar)(X_i - \xbar)\T \hSigma\inv\right] \int_\Ist C_T(S_i(u), S_i(v)) \bfp(u) \d u \right\} \\
&\hspace{1cm}  \leq \frac{C_1\hat{W}_n}{n^2} \mathrm{tr}\left(\hSigma\inv \son (X_i - \xbar)(X_i - \xbar)\T\right)  = \frac{pD\hat{W}_n}{n}.
\end{split}
\]
Thus,
\begin{equation*}
\PX\left(\left|\int_\Ist (\hat{\rho}(u) - \rho(u)) \hSigma \rho(u) \bfp(u) \d u \right| > \frac{W}{8}\right) \leq \frac{64pC_1\hat{W}_n}{nW^2} = o\left(\frac{\hat{W}_n}{W}\right),
\end{equation*}
where the $o(\cdot)$ term is uniform over $\mc{G} \in \mk{G}_{A,n}$ as $nW \rightarrow \infty$ uniformly.  Hence, choosing $\eta > 0$ such that $\eta < \epsilon / 3$, for large $n$ we will have
\begin{equation}
\label{eq: V8}
\begin{split}
&P\left(\sup_{m\geq n} \PX\left(\left|\int_\Ist (\hat{\rho}(u) - \rho(u))\T \hSigma \rho(u) \bfp(u) \d u \right| > \frac{W}{8}\right) > \frac{\epsilon}{2}\right) \leq P\left(\sup_{m \geq n} \hat{W}_m > \frac{W\epsilon}{2\eta}\right) \\
& \hspace{1cm} \leq P\left(\sup_{m \geq n} |\hat{W}_m - W| > \frac{W}{2}\right). 
\end{split}
\end{equation}

Next, write
$$
\hat{W}_n - W = \int_{u: \rho(u) \neq 0} \rho(u)\T\Sigma\rho(u)\left[\frac{\rho(u)\T\hSigma \rho(u)}{\rho(u)\T\Sigma \rho(u)} - 1\right] \bfp(u) \d u.
$$
An elementary calculation combined with Lemma~4.3 of \ci{bosq:00} gives
$$
\sup_{u: \rho(u) \neq 0} \left|\frac{\rho(u)\T\hSigma \rho(u)}{\rho(u)\T\Sigma \rho(u)} - 1\right| \leq C \normF{\hSigma - \Sigma},
$$
for some $C$ depending only on the distribution of $X.$  As a consequence, $|\hat{W}_n - W| \leq CW\normF{\hSigma - \Sigma}.$  Then, for any $\epsilon > 0,$ choose $\eta$ and $n$ so that \eqref{eq: V8} holds. For such $n,$
\[
\begin{split}
&P\left(\sup_{m \geq n} \PX\left(\Delta_m > \frac{W}{2}\right) > \epsilon\right) \\
&\hspace{0.5cm} \leq P\left(\sup_{m \geq n}\PX\left(\left|\int_\Ist (\hat{\rho}(u) - \rho(u)) \hSigma \rho(u) \bfp(u) \d u \right| > \frac{W}{8}\right) > \frac{\epsilon}{2}\right)\\
&\hspace{1cm} + P\left(\sup_{m \geq n} \mathbf{1}\left(|\hat{W}_m - W| > \frac{V}{2}\right) > \frac{\epsilon}{2}\right) \\
&\hspace{0.5cm} \leq 2P\left(\sup_{m \geq n} |\hat{W}_m - W| > \frac{W}{2}\right) \leq 2P\left(\sup_{m \geq n} \normF{\hSigma - \Sigma} > \frac{1}{2}\right) = o(1)
\end{split}
\]
uniformly in $\mk{G}_{A,n}$ since the distribution of $X$ is fixed and $\normF{\hSigma - \Sigma} \rightarrow 0$ almost surely.

Next, consider $j = 2$.  For any $\delta > 0,$ let $\{u_k\}_{k = 1}^K$ be a $\delta$-covering of $\Ist,$ where $K = O(\delta\inv).$ Observe that, by (G3) and (G4), $$\tilde{x}\T \pi'\circ \bQp(t) = \qp(x,t)/\bqp(t) \geq M_1M_2\inv =: M_3 > 0$$ for almost all $x$ and all $t \in (0,1),$ so that whenever
$$
\max_{1\leq i \leq n} \sup_{u \in \Ist} |\tilde{X}_i\T(\tilde{\pi}'(u) - \pi'(u))| < M_3,
$$
the event $A_n$ in \eqref{eq: An} will hold.  Thus, we obtain the bound
\[
\begin{split}
\PX(A_n^c) &\leq \PX\left(\max_{1 \leq i\leq n} \max_k |\tilde{X}_i\T(\tilde{\pi}'(u_k) - \pi'(u_k))| \geq \frac{M_3}{2}\right) \\
&\hspace{1cm} + \PX\left(\max_{1 \leq i \leq n} \sup_{|u-v| < \delta} |\tilde{X}_i\T(\tilde{\pi}'(u) - \tilde{\pi}'(v))| \geq \frac{M_3}{2}\right).
\end{split}
\]

Let $C_2 = \sup_{\mc{G} \in \mk{G}\ast} \sup_{u \in \R} C_T^{(1,1)}(u,u)<\infty.$ By (G1), there is a constant $C_3 <\infty$ such that $\sup_{u \in \Ist}\normE{\pi'(u)} < C_3$ uniformly in $\mk{G}\ast.$  Since $S_i(u) = \tilde{X}_i\T\pi(u),$ $|S_i'(u)| \leq C_3\normE{\tilde{X}_i}.$   Furthermore, $\max_{1 \leq i \leq n} \normE{\tilde{X}_i} = o(n^{1/4})$ and $\normF{\hLambda\inv} = O(1)$ almost surely.  Then, for any $\eta > 0$ and some $B >0,$ for large enough $n$ we will have
\begin{equation}
\label{eq: An_bd1}
\begin{split}
&\PX\left(\max_{1 \leq i \leq n} \max_k |\tilde{X}_i\T(\tilde{\pi}'(u_k) - \pi'(u_k)) \geq \frac{M_3}{2}\right) \leq 4M_3^{-2} K \max_k \EX\left[\max_{1 \leq i \leq n} |\tilde{X}_i\T(\tilde{\pi}'(u_k) - \pi'(u_k))|^2\right] \\
&\hspace{1cm} \leq 4M_3^{-2} K\max_{k} \mathrm{tr}\left[\EX\left(\hLambda (\tilde{\pi}'(u_k) - \pi'(u_k))(\tilde{\pi}'(u_k) - \pi'(u_k))\T \hLambda\right)\right] \left(\max_{1 \leq i \leq n} \normE{\hLambda\inv \tilde{X}_i}^2\right) \\
&\hspace{1cm} = O\left(KB^2\eta^2n^{1/2}\right) \max_k \left[\frac{1}{n} \son \normE{\tilde{X}_i}^2 C_T^{(1,1)}(S_i(u_k), S_i(u_k))(S_i'(u_k))^2\right] \\
&\hspace{1cm} = O\left(C_3^2C_2'K \eta^2n^{1/2}\right)\left(\frac{1}{n^2}\son \normE{\tilde{X}_i}^4\right)\\
&\hspace{1cm} = O\left(\eta^2\delta\inv n^{-1/2}\right)
\end{split}
\end{equation}
almost surely, where the $O(\cdot)$ term depends only on the distribution of $X,$ and is thus uniform over $\mk{G}\ast.$  

Next, define 
$$
\nu_n = \frac{1}{n} \sum_{i = 1}^n \normE{\tilde{X}_i}^2\left\{E(L_i|X_i) \normE{\tilde{X}_i} + E\left(\sup_{u \in \R} |T_i'(u)| \, | X_i\right)\right\}
$$
By (G1), there exists $C_4<\infty$ such that $\sup_{u \in \Ist} \normE{\pi''(u)} < C_4$ uniformly in $\mk{G}\ast.$ Then, using \eqref{eq: Rbounds},
$$
\normE{\tilde{\pi}'(u) - \tilde{\pi}'(v)} \leq C_4(1+C_2)\normF{\hLambda\inv}\nu_n|u-v|.
$$
Hence, for large enough $n,$
\begin{equation}
\label{eq: An_bd2}
\begin{split}
\PX\left(\max_{1 \leq i \leq n} \sup_{|u-v| < \delta} |\tilde{X}_i\T(\tilde{\pi}'(u) - \tilde{\pi}'(v))| \geq \frac{M_3}{2}\right) &\leq  2C_4(1+C_2)M_3\inv \delta\normF{\hLambda\inv}\left(\max_{1 \leq i \leq n} \normE{\tilde{X}_i}\right)\nu_n \\
&= O(\delta \eta n^{1/4} \nu_n),
\end{split}
\end{equation}
where the $O(\cdot)$ terms is again uniform over $\mk{G}\ast.$  

Finally, by (G2) and Proposition A.5.1 of \ci{well:96},
$$
\limsup_{M \rightarrow \infty} \limsup_{n \rightarrow \infty} \sup_{\mc{G} \in \mk{G}\ast}  P\left( \sup_{m \geq n} \nu_m > M\right) \rightarrow 0.
$$
Setting $\delta = n^{-1/4},$ \eqref{eq: An_bd1} and \eqref{eq: An_bd2} imply that, for any $M > 0$ and for large $n,$
\[
\begin{split}
P\left(\sup_{m \geq n} \PX\left(A_n^c\right) > \epsilon\right) &\leq P\left(\sup_{m \geq n} O\left(\eta^2m^{-1/4}\right) > \frac{\epsilon}{2}\right) \\
&\hspace{1cm} P\left(\sup_{m \geq n} O(M\eta) > \frac{\epsilon}{2}\right) + P\left(\sup_{m \geq n} \nu_m > M\right).
\end{split}
\]
The last term can be made arbitrarily small uniformly in $\mk{G}\ast$ by choosing $M$ large.  Then, one can choose $\eta$ small and $n$ large so that the first and second terms vanish for any $\mc{G} \in \mk{G}\ast.$ Hence $r_{2n} = o(1).$

Lastly, consider $j = 3$ and let $R > 0$ be arbitrary.  Let $\hat{H}$ be the cdf of $W_n = \sum_{j = 1}^{\infty} \hat{\lambda}_j \omega_j$ (recall that $\omega_j \overset{\mathrm{iid}}{\sim} \chi^2_p$) conditional on the data, so that $\hat{H}(\hat{b}_\alpha^G) = 1 - \alpha.$  Then $\hat{b}_\alpha^G> R$ implies that $\hat{H}(R) < 1 - \alpha.$  Denote by $P_\omega$ the probability measure of the $\omega_j$, treating the observed data as constant.  Then
\[
1 - \hat{H}(R)=P_\omega\left(\sum_{j = 1}^{\infty} \hlambda_j V_j > R\right) \leq pR\inv \sum_{j = 1}^{\infty} \hlambda_j \leq pR\inv \i01 \hat{C}_Q(t,t) \d t,
\]
where we have used the identity $\sum_{j = 1}^\infty \hlambda_j = \i01 \hat{C}_Q(t,t)\d t.$  

Set $\tilde{C}_Q(t,t) = n\inv \son (Q_i(t) - \tilde{Q}_\oplus(X_i,t))^2,$ so that $\hat{C}_Q(t,t) = \tilde{C}_Q(t,t)$ whenever $A_n$ holds.  Whenever $\normF{\hLambda\inv}^2 \leq B_1$ and $n\inv \son \normE{\tilde{X}_i}^2 \leq B_2$, using (G1) one can derive the bound
\begin{equation}
\EX\left[\i01 \tilde{C}_Q(t,t) \d t\right] \leq D\left[\frac{n-2}{n} + \frac{B_1(1 + B_2^2)}{n}\right] = O(1).
\end{equation}
uniformly over $\mk{G}\ast.$  Thus, for large $n,$ the above is bounded by some constant $C$ so that
\begin{equation}
\label{eq: balpha_bd}
\begin{split}
&P\left(\sup_{m \geq n} \PX(\hat{b}_\alpha^G > R) > \epsilon\right) \\
&\hspace{0.5cm} \leq P\left(\sup_{m \geq n} \PX\left(\i01 \tilde{C}_Q(t,t) \d t > Rp\inv \alpha\right) > \epsilon\right) + P\left(\sup_{m \geq n} \PX(A_m^c)\right) \\
&\hspace{0.5cm} \leq \mathbf{1}\left(C > Rp\inv\alpha\right) + P\left(\sup_{m \geq n} \normF{\hLambda\inv}^2 > B_1\right) + P\left(\sup_{m \geq n} \frac{1}{n}\son \normE{\tilde{X}_i}^2 > B_2\right) + r_{2n},
\end{split}
\end{equation}
which is uniform over $\mk{G}\ast$ since the distribution of $X$ is fixed.  By choosing $R$ large, the first term vanishes.  The second and third terms can also be made small, uniformly in $\mk{G}\ast,$ since $\normF{\hLambda\inv}$ and $n\inv \son \normE{\tilde{X}_i}^2$ are both $O(1)$ almost surely.  Thus, $r_{3n} \rightarrow 0.$
\end{proof}


\begin{Lemma}
\label{lma: hadamard}
Let $\mathbb{E} = L^\infty\zo \times L^\infty(0,1),$ $\mathbb{E}_0 = C[0,1]\times C_b(0,1),$ and
$$
\mathbb{E}_\phi = \left\{(h_1,h_2) \in \mathbb{E}: \inf_{t \in (0,1)} h_2(t) > 0, \, h_1 \text{\emph{ strictly increasing, }} \I_x^\delta \subset (h_1(0), h_1(0)).\right\}
$$
Define the map $\phi: \mathbb{E}_\phi \rightarrow L^\infty(\I_x^\delta)$ as $\phi(h_1,h_2) = \frac{1}{h_2 \circ h_1\inv}.$  Then, provided $\sup_{u \in \Ist} |{\bfp}'(u)| < \infty,$ $\phi$ is Hadamard differentiable at $(\Qp(x), \qp(x)) \in \mathbb{E}_\phi,$ tangentially to $\mathbb{E}_0,$ with
\begin{equation}
\label{eq: hadamard}
\phi'_{(\Qp(x),\qp(x))}(h_1,h_2) = \left[\frac{\partial}{\partial u} \fp(x) \right] h_1 \circ \Fp(x) -\fp(x)^2 h_2 \circ \Fp(x).
\end{equation}
\end{Lemma}

\begin{proof}
First, define $\mathbb{F}_\phi = \{h \in L^\infty(\I_x^\delta):\, \inf_{u \in \I_x^\delta} h(u) > 0\}$ and $\mathbb{F}_0 = C(\I_x^\delta).$  Define $\phi_1:\mathbb{E}_\phi \rightarrow \mathbb{F}_\phi$ by $\phi_1(h_1,h_2) = h_2\circ h_1\inv,$ and $\phi_2: \mathbb{F}_\phi \rightarrow L^\infty(\I_x^\delta)$ by $\phi_2(h) = 1/h.$  We will show that each map is Hadamard differentiable, then apply the chain rule.

For $(h_2,h_2) \in \mathbb{E}_0,$ let $r_k$ be a real-valued sequence converging to zero.  Furthermore, define a sequence $(h_{1k}, h_{2k}) \in \mathbb{E}$ such that $h_{jk}$ converges uniformly to $h_j,$ and $$(g_{1k}, g_{2k}) = (\Qp(x), \qp(x)) + r_k(h_{1k}, h_{2k}) \in \mathbb{E}_\phi$$ for all $k.$  Then, for any $u \in \I_x^\delta,$ there exists $t_k(u)$ between $g_{1k}\inv(u)$ and $\Fp(x,u)$ such that
\begin{equation}
\label{eq: diff1}
\phi_1(g_{1k}, g_{2k})(u) - \phi_1(\Qp(x), \qp(x))(u) = r_kh_{2k}\circ g_{1k}\inv(u) + \frac{\partial}{\partial t} \qp(x,t_k(u))(g_{1k}\inv(u) - \Fp(x,u)).
\end{equation}
Define $t_k = g_{1k}\inv(u).$  Then there exists $u_k$ between $\Qp(x,t_k)$ and $u$ such that 
$$
g_{1k}\inv(u) - \Fp(x,u) = -\frac{r_k}{\qp(x) \circ \Fp(x, u_k)}h_{1k}(t_k).
$$
Since $\inf_{t \in (0,1)} \bqp(x,t) > 0,$ we have $\sup_{u \in \I_x^\delta} |g_{1k}\inv(u) - \Fp(x,u)| \rightarrow 0.$  As a consequence,
\begin{equation}
\label{eq: h2k}
\sup_{u \in \I_x^\delta} |h_{2k}\circ g_{1k}\inv(u) - h_2\circ \Fp(x,u)| \leq \sup_{t \in \zo} |h_{2k}(u) - h_2(u)| +  \sup_{u \in \I_x^\delta} |h_2\circ g_{1k}\inv(u) - h_2\circ \Fp(x,u)|.
\end{equation}
The first term goes to zero by assumption on $h_{2k},$ while the second converges to zero since $h_2$ is uniformly continuous on any compact subset of $[0,1].$  

Next, since $\bqp$ is continuously differentiable, it follows that $(\partial / \partial t) \qp(x)$ is continuous on $(0,1).$  Then, since $|t_k(u) - \Fp(x,u)| \leq |g_{1k}\inv(u) - \Fp(x,u)|,$ we have that $$(\partial / \partial t) \qp(x, t_k(u)) \rightarrow \frac{\partial}{\partial t} \qp(x, \Fp(x,u)) = -\frac{(\partial / \partial u) \fp(x,u)}{\fp(x,u)^3}$$
uniformly. 

Finally, by continuity of $h_1$ and $\qp(x),$ one can show that $$r_k\inv(g_{1k}\inv(u) - \Fp(x,u)) \rightarrow \frac{h_1\circ \Fp(x,u)}{\qp(x)\circ \Fp(x,u)} =\fp(x,u)h_1\circ \Fp(x,u)$$ uniformly.   Plugging these results into \eqref{eq: diff1}, we find that $\phi_1$ is differentiable at $(\Qp(x), \qp(x)) \in \mathbb{E}_\phi$, tangentially to $\mathbb{E}_0,$ with
\begin{equation}
\label{eq: phi1HD}
\begin{split}
\phi_{1,(\Qp(x), \qp(x))}'(h_1,h_2)& = \lim_{k \rightarrow \infty} \frac{\phi_1(g_{1k}, g_{2k}) - \phi_1(\Qp(x), \qp(x))}{r_k} \\
&= \left[h_2 + \frac{(\partial /\partial t)\qp(x)}{\qp(x)} h_1\right] \circ \Fp(x) \\
&= h_2 \circ \Fp(x) - \frac{(\partial / \partial u) \fp(x)}{\fp(x)^2}h_1\circ \Fp(x).
\end{split}
\end{equation}

Using similar methods, one can show that $\phi_2$ is Hadamard differentiable at $\phi_1(\Qp(x), \qp(x)) = \qp(x)\circ \Fp(x) \in \mathbb{F}_\phi,$ tangentially to $\mathbb{F}_0$, with
\begin{equation}
\label{eq: Phi2HD}
\phi_{2, \qp(x) \circ \Fp(x)}'(h) = -h\fp(x)^2.
\end{equation}
Since $\phi_1(\mathbb{E}_0) \subset \mathbb{F}_0,$ the result holds by the chain rule.

\end{proof}

\end{document}